\documentclass[12pt,reqno]{amsart}
\usepackage{graphicx, amsmath, amssymb}
\numberwithin{equation}{section}

\newtheorem{Thm}{Theorem}[section]

\newtheorem{Prop}[Thm]{Proposition}
\newtheorem{Cor}[Thm]{Corollary}
\newtheorem{Lem}[Thm]{Lemma}

\def\bB {\mathbf{B}}
\def\bbb{\mathbf{b}}

\def\bE {\mathbf{E}}
\def\bK {\mathbf{K}}
\def\bN {\mathbf{N}}
\def\bQ {\mathbf{Q}}
\def\bR {\mathbf{R}}
\def\bS {\mathbf{S}}
\def\bT {\mathbf{T}}
\def\bZ {\mathbf{Z}}

\def\cA {\mathcal{A}}
\def\cB {\mathcal{B}}
\def\cC {\mathcal{C}}

\def\cF {\mathcal{F}}

\def\cK {\mathcal{K}}
\def\cL {\mathcal{L}}

\def\cQ {\mathcal{Q}}
\def\cR {\mathcal{R}}

\def\cZ {\mathcal{Z}}

\def\a {{\alpha}}
\def\b {{\beta}}
\def\g {{\gamma}}
\def\de {{\delta}}
\def\eps {{\varepsilon}}
\def\th {{\theta}}

\def\l {{\lambda}}
\def\si {{\sigma}}
\def\Si {{\Sigma}}
\def\om {{\omega}}
\def\Om {{\Omega}}

\def\rstr {{\big |}}
\def\indc {{\bf 1}}

\def\la {\langle}
\def\ra {\rangle}

\def\d {{\partial}}
\def\grad {{\nabla}}
\def\Dlt {{\Delta}}

\newcommand{\be}{\begin{equation}}
\newcommand{\ee}{\end{equation}}
\newcommand{\bmat}{\begin{matrix}}
\newcommand{\emat}{\end{matrix}}
\newcommand{\ba}{\begin{aligned}}
\newcommand{\ea}{\end{aligned}}

\begin{document}

\title[Periodic Lorentz Gas]
{Recent Results\\ on the Periodic Lorentz Gas}

\author[F. Golse]{Fran\c cois Golse}
\address[F. G.]
{Ecole polytechnique\\
Centre de Math\'ematiques L. Schwartz\\
F91128 Palaiseau Cedex}
\email{golse@math.polytechnique.fr}



\begin{abstract}
The Drude-Lorentz model for the motion of electrons in a solid is a classical model in statistical 
mechanics, where electrons are represented as point particles bouncing on a fixed system of
obstacles (the atoms in the solid). Under some appropriate scaling assumption --- known as the
Boltzmann-Grad scaling by analogy with the kinetic theory of rarefied gases --- this system can
be described in some limit by a linear Boltzmann equation, assuming that the configuration of 
obstacles is random [G. Gallavotti, [Phys. Rev. (2) {\bf 185} (1969), 308]). The case of a periodic
configuration of obstacles (like atoms in a crystal) leads to a completely different limiting dynamics.
These lecture notes review several results on this problem obtained in the past decade as joint
work with J. Bourgain, E. Caglioti and B. Wennberg.
\end{abstract}

\maketitle

\tableofcontents


\section*{Introduction: from particle dynamics to kinetic models}


The kinetic theory of gases was proposed by J. Clerk Maxwell \cite{Maxwell1852, Maxwell1866}
and L. Boltzmann \cite{Boltzmann1872} in the second half of the XIXth century. Because the 
existence of atoms, on which kinetic theory rested, remained controversial for some time, it was
not until many years later, in the XXth century, that the tools of kinetic theory became of common
use in various branches of physics such as neutron transport, radiative transfer, plasma and 
semiconductor physics... 

Besides, the arguments which Maxwell and Boltzmann used in writing what is now known as the
``Boltzmann collision integral" were far from rigorous --- at least from the mathematical viewpoint. 
As a matter of fact, the Boltzmann equation itself was studied by some of the most distinguished 
mathematicians of the XXth century --- such as Hilbert and Carleman --- before there were any 
serious attempt at deriving this equation from first principles (i.e. molecular dynamics.) Whether 
the Boltzmann equation itself was viewed as a fundamental equation of gas dynamics, or as 
some approximate equation valid in some well identified limit is not very clear in the first works 
on the subject --- including Maxwell's and Boltzmann's.

It seems that the first systematic discussion of the validity of the Boltzmann equation viewed as 
some limit of molecular dynamics --- i.e. the free motion of a large number of small balls subject to 
binary, short range interaction, for instance elastic collisions --- goes back to the work of H. Grad
\cite{Grad1958}.
In 1975, O.E. Lanford gave the first rigorous derivation \cite{Lanford1975} of the Boltzmann equation 
from molecular dynamics --- his result proved the validity of the Boltzmann equation for a very short 
time of the order of a fraction of the reciprocal collision frequency. (One should also mention an 
earlier, ``formal derivation" by C. Cercignani \cite{Cercignani1972} of the Boltzmann equation for 
a hard sphere gas, which considerably clarified the mathematical formulation of the problem.)
Shortly  after Lanford's derivation of the Boltzmann equation, R. Illner and M. Pulvirenti managed
to extend the validity of his result for all positive times, for initial data corresponding with a very 
rarefied cloud of gas molecules \cite{IllnerPulvirenti1986}.

An important assumption made in Boltzmann's attempt at justifying the equation bearing his name 
is the ``Stosszahlansatz", to the effect that particle pairs just about to collide are uncorrelated. 
Lanford's argument indirectly established the validity of Boltzmann's assumption, at least on very 
short time intervals.

\smallskip
In applications of kinetic theory other than rarefied gas dynamics, one may face the situation where 
the analogue of the Boltzmann equation for monatomic gases is linear, instead of quadratic. The 
linear Boltzmann equation is encountered for instance in neutron transport, or in some models in 
radiative transfer. It usually describes a situation where particles interact with some background 
medium --- such as neutrons with the atoms of some fissile material, or photons subject to scattering 
processes (Rayleigh or Thomson scattering) in a gas or a plasma.

In some situations leading to a linear Boltzmann equation, one has to think of two families of particles: 
the moving particles whose phase space density satisfies the linear Boltzmann equation, and the 
background medium that can be viewed as a family of fixed particles of a different type. For instance, 
one can think of the moving particles as being light particles, whereas the fixed particles can be 
viewed as infinitely heavier, and therefore unaffected by elastic collisions with the light particles. 
Before Lanford's fundamental paper, an important --- unfortunately unpublished --- preprint by 
G. Gallavotti \cite{Gallavotti1972} provided  a rigorous derivation of the linear Boltzmann equation 
assuming that the background medium consists of fixed, independent like hard spheres whose 
centers are distributed in the Euclidian space under Poisson's law. Gallavotti's argument already 
possessed some of the most remarkable features in Lanford's proof, and therefore must be regarded 
as an essential step in the understanding of kinetic theory.

However, Boltzmann's Stosszahlansatz becomes questionable in this kind of situation involving light 
and heavy particles, as potential correlations among heavy particles may influence the light particle 
dynamics. Gallavotti's assumption of a background medium consisting of independent hard spheres 
excluded this this possibility. Yet, strongly correlated background media are equally natural,
and should also be considered.

The periodic Lorentz gas discussed in these notes is one example of this type of situation. Assuming 
that heavy particles are located at the vertices of some lattice in the Euclidian space clearly introduces 
about the maximum amount of correlation between these heavy particles. This periodicity assumption 
entails a dramatic change in the structure of the equation that one obtains under the same scaling limit 
that would otherwise lead to a linear Boltzmann equation.

Therefore, studying the periodic Lorentz gas can be viewed as one way of testing the limits of the classical 
concepts of the kinetic theory of gases.

\bigskip
\noindent
\textbf{Acknowledgements.} 

Most of the material presented in these lectures is the result of collaboration with several authors: 
J. Bourgain, E. Caglioti, H.S. Dumas, L. Dumas and B. Wennberg, whom I wish to thank for sharing
my interest for this problem. I am also grateful to C. Boldighrini and G. Gallavotti for illuminating
discussions on this subject.


\section{The Lorentz kinetic theory for electrons}


In the early 1900's, P. Drude \cite{Drude1900} and H. Lorentz \cite{Lorentz1905} independently 
proposed to describe the motion of electrons in metals by the methods of kinetic theory. One 
should keep in mind that the kinetic theory of gases was by then a relatively new subject: the 
Boltzmann equation for monatomic gases  appeared for the first time in the papers of J. Clerk 
Maxwell \cite{Maxwell1866} and L. Boltzmann \cite{Boltzmann1872}. Likewise, the existence of 
electrons had been established shortly before, in 1897 by J.J. Thomson.

\begin{figure}

\includegraphics[width=6.0cm]{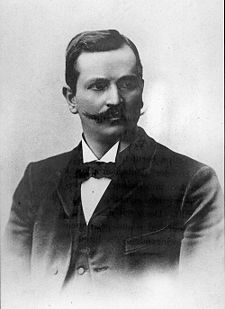}\includegraphics[width=5.8cm]{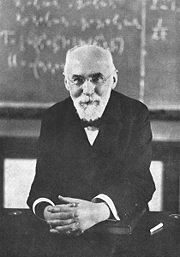}

\caption{Left: Paul Drude (1863-1906); right: Hendrik Antoon Lorentz (1853-1928)}

\end{figure}

The basic assumptions made by H. Lorentz in his paper \cite{Lorentz1905} can be summarized as 
follows.

First, the population of electrons is thought of as a gas of point particles described by its phase-space 
density $f\equiv f(t,x,v)$, that is the density of electrons at the position $x$ with velocity $v$ at time $t$.

Electron-electron collisions are neglected in the physical regime considered in the Lorentz kinetic 
model --- on the contrary, in the classical kinetic theory of gases, collisions between molecules are 
important as they account for momentum and heat transfer.

However, the Lorentz kinetic theory takes into account collisions between electrons and the 
surrounding metallic atoms. These collisions are viewed as simple, elastic hard sphere collisions.

Since electron-electron collisions are neglected in the Lorentz model, the equation governing the 
electron phase-space density $f$ is linear. This is at variance with the classical Boltzmann equation, 
which is quadratic because only binary collisions involving pairs of molecules are considered in the 
kinetic theory of gases.

With the simple assumptions above, H.  Lorentz arrived at the following equation for the phase-space 
density of electrons $f\equiv f(t,x,v)$:
$$
(\d_t+v\cdot\grad_x+\tfrac1m F(t,x)\cdot\grad_v)f(t,x,v)
	=N_{at}r_{at}^2|v|\cC(f)(t,x,v)\,.
$$

In this equation, $\cC$ is the Lorentz collision integral, which acts on the only variable $v$ in the 
phase-space density $f$. In other words, for each continuous function $\phi\equiv \phi(v)$, one 
has
$$
\cC(\phi)(v)=\int_{|\om|=1\atop\om\cdot v>0}
	\bigl(\phi(v-2(v\cdot\om)\om)-\phi(v)\bigr)\cos(v,\om)d\om\,,
$$
and the notation
$$
\cC(f)(t,x,v)\hbox{ designates }\cC(f(t,x,\cdot))(v)\,.
$$
The other parameters involved in the Lorentz equation are the mass $m$ of the electron, and $N_{at}$, 
$r_{at}$ respectively the density and radius of metallic atoms. The vector field $F\equiv F(t,x)$ is the 
electric force. In the Lorentz model, the self-consistent electric force --- i.e. the electric force created by 
the electrons themselves --- is neglected, so that $F$ take into account the only effect of an applied 
electric field (if any). Roughly speaking, the self consistent electric field is linear  in $f$, so that its 
contribution to the term $F\cdot\grad_vf$ would be quadratic in $f$, as would be any collision integral 
accounting for electron-electron collisions. Therefore, neglecting electron-electron collisions and the
self-consistent electric field are both in accordance with assuming that $f\ll 1$.

\smallskip
The line of reasoning used by H. Lorentz to arrive at the kinetic equations above is based on the 
postulate that the motion of electrons in a metal can be adequately represented by a simple mechanical 
model --- a collisionless gas of point particles bouncing on a system of fixed, large spherical obstacles 
that represent the metallic atoms. Even with the considerable simplification in this model, the argument 
sketched in the article \cite{Lorentz1905} is little more than a formal analogy with Boltzmann's derivation 
of the equation now bearing his name. 

This suggests the mathematical problem, of deriving the Lorentz kinetic equation from a microscopic, 
purely mechanical particle model. Thus, we consider a gas of point particles (the electrons) moving 
in a system of fixed spherical obstacles (the metallic atoms). We assume that collisions between the 
electrons and the metallic atoms are perfectly elastic, so that, upon colliding with an obstacle, each 
point particle is specularly reflected on the surface of that obstacle.

Undoubtedly, the most interesting part of the Lorentz kinetic equation is the collision integral which
does not seem to involve $F$. Therefore we henceforth assume for the sake of simplicity that there is 
no applied electric field, so that 
$$
F(t,x)\equiv 0\,.
$$

In that case, electrons are not accelerated between successive collisions with the metallic atoms, so 
that the microscopic model to be considered is a simple, dispersing billiard system --- also called a 
Sinai billiard. In that model, electrons are point particles moving at a constant speed along rectilinear 
trajectories in a system of fixed spherical obstacles, and specularly reflected at the surface of the 
obstacles.

\begin{figure}

\begin{center}
\includegraphics[width=6.0cm]{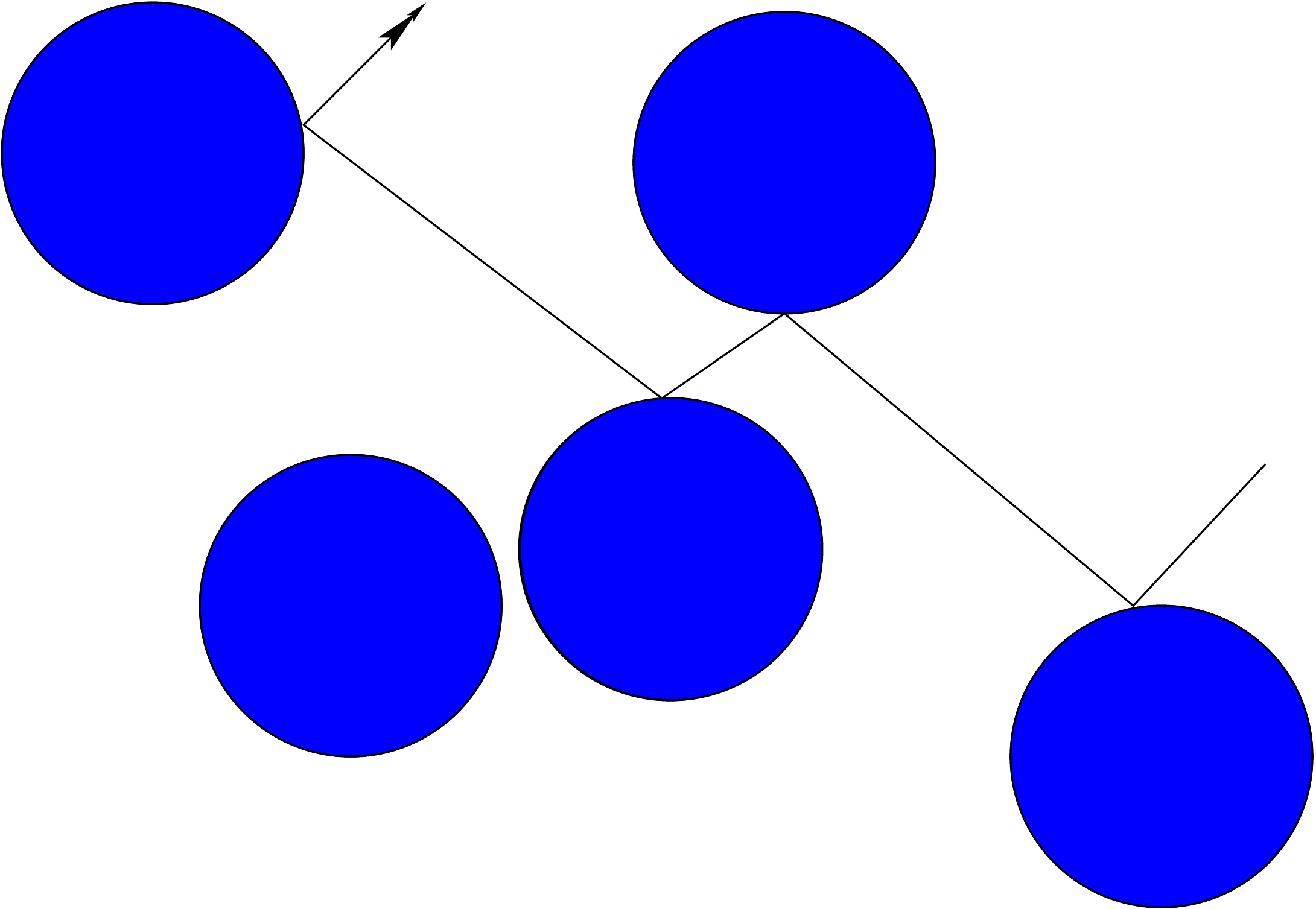}
\end{center}

\caption{The Lorentz gas: a particle path}

\end{figure}

More than 100 years have elapsed since this simple mechanical model was proposed by P. Drude 
and H. Lorentz, and today we know that the motion of electrons in a metal is a much more complicated 
physical phenomenon whose description involves quantum effects. 

Yet the Lorentz gas is an important object of study in nonequilibrium satistical mechanics, and there 
is a very significant amount of literature on that topic --- see for instance \cite{SzaszEncyclo} and the 
references therein. 

The first rigorous derivation of the Lorentz kinetic equation is due to G. Gallavotti \cite{Gallavotti1969,
Gallavotti1972}, who derived it from from a billiard system consisting of randomly (Poisson) distributed 
obstacles, possibly overlapping, considered in some scaling limit --- the Boltzmann-Grad limit, whose 
definition will be given (and discussed) below. Slightly more general, random distributions of obstacles 
were later considered by H. Spohn in \cite{Spohn1978}. 

While Gallavotti's theorem bears on the convergence of the mean electron density (averaging over 
obstacle configurations), C. Boldrighini, L. Bunimovich and Ya. Sinai \cite{BoldriBuniSinai1983} later 
succeeded in proving the almost sure convergence (i.e. for a.e. obstacle configuration) of the electron 
density to the solution of the Lorentz kinetic equation.

In any case, none of the results above says anything on the case of a periodic distribution of obstacles. 
As we shall see, the periodic case is of a completely different nature --- and leads to a very different 
limiting equation, involving  a phase-space different from the one considered by H. Lorentz --- i.e.
$\bR^2\times\bS^1$ --- on which the Lorentz kinetic equation is posed. 

The periodic Lorentz gas is at the origin of many challenging mathematical problems. 
For instance, in the late 1970s, L. Bunimovich and Ya. Sinai studied the periodic Lorentz gas in a 
scaling limit different from the Boltzmann-Grad limit studied in the present paper. In \cite{BuniSinai1980}, 
they showed that the classical Brownian motion is the limiting dynamics of the Lorentz gas under 
that scaling assumption --- their work was later extended with N. Chernov: see \cite{BunSinaiChern1991}. 
This result is indeed a major achievement in nonequilibrium statistical mechanics, as it provides
an example of an irreversible dynamics (the heat equation associated with the classical Brownian motion) 
that is derived from a reversible one (the Lorentz gas dynamics).


\section{The Lorentz gas in the Boltzmann-Grad limit\\ with a Poisson distribution of obstacles}


Before discussing the Boltzmann-Grad limit of the periodic Lorentz gas, we first give a brief description 
of Gallavotti's result \cite{Gallavotti1969, Gallavotti1972} for the case of a Poisson distribution of 
independent, and therefore possibly overlapping obstacles. As we shall see, Gallavotti's argument is 
in some sense fairly elementary, and yet brilliant. 

First we define the notion of a Poisson distribution of obstacles. Henceforth, for the sake of simplicity,
we assume a $2$-dimensional setting.

The obstacles (metallic atoms) are disks of radius $r$ in the Euclidian plane $\bR^2$, centered at 
$c_1,c_2,\ldots,c_j,\ldots\in\bR^2$. Henceforth, we denote by
$$
\{c\}=\{c_1,c_2,\ldots,c_j,\ldots\}=\hbox{ a configuration of obstacle centers.}
$$

We further assume that the configurations of obstacle centers $\{c\}$ are distributed under Poisson's 
law with parameter $n$, meaning that
$$
\hbox{Prob}(\{\{c\}\,|\,\#(A\cap\{c\})=p\})=e^{-n|A|}\frac{(n|A|)^p}{p!}\,,
$$
where $|A|$ denotes the surface, i.e. the $2$-dimensional Lebesgue measure of a measurable subset 
$A$ of the Euclidian plane $\bR^2$.

This prescription defines a probability on countable subsets of the Euclidian plane $\bR^2$.

Obstacles may overlap: in other words, configurations $\{c\} $ such that
$$
\hbox{for some $j\not=k\in\{1,2,\ldots\}$, one has }|c_i-c_j|<2r 
$$
are not excluded. Indeed, excluding overlapping obstacles means rejecting obstacles configurations
$\{c\}$ such that $|c_i-c_j|\le 2r$ for some $i,j\in\bN$. In other words, $\hbox{Prob}(d\{c\})$ is replaced
with
$$
\frac1Z\prod_{i>j\ge 0}\indc_{|c_i-c_j|>2r}\hbox{Prob}(d\{c\})\,,
$$
(where $Z>0$ is a normalizing coefficient.) Since the term
$$
\prod_{i>j\ge 0}\indc_{|c_i-c_j|>2r}\hbox{ is not of the form }\prod_{k\ge 0}\phi_k(c_k)\,,
$$
the obstacles are no longer independent under this new probability measure.

Next we define the billiard flow in a given obstacle configuration $\{c\}$. This definition is self-evident, 
and we give it for the sake of completeness, as well as in order to introduce the notation.

Given a countable subset $\{c\}$ of the Euclidian plane $\bR^2$, the billiard flow in the system of obstacles 
defined by $\{c\}$ is the family of mappings
$$
(X(t;\cdot,\cdot,\{c\}),V(t;\cdot,\cdot,\{c\})):
	\,\left(\bR^2\setminus\bigcup_{j\ge 1}B(c_j,r)\right)\times\bS^1\circlearrowright
$$
defined by the following prescription. 

Whenever the position $X$ of a particle lies outside the surface of any obstacle, that particle moves at 
unit speed along a rectilinear path: 
$$
\begin{aligned}
\dot{X}(t;x,v,\{c\} )&=V(t;x,v,\{c\} )\,,
\\
\dot{V}(t;x,v,\{c\} )&=0\,,\qquad\hbox{ whenever }|X(t;x,v,\{c\} )-c_i|>r\hbox{ for all }i\,,
\end{aligned}
$$
and, in case of a collision with the $i$-th obstacle, is specularly reflected on the surface of that obstacle 
at the point of impingement, meaning that
$$
\begin{aligned}
X(t+0;x,v,\{c\} )&=X(t-0;x,v,\{c\})\in\d B(c_i,r)\,,
\\
V(t+0;x,v,\{c\} )&=\cR\left[\frac{X(t;x,v,\{c\})-c_i}{r}\right]V(t-0;x,v,\{c\} )\,,
\end{aligned}
$$
where $\cR[\om]$ denotes the reflection with respect to the line $(\bR\om)^\bot$:
$$
\cR[\om]v=v-2(\om\cdot v)\om\,,\quad|\om|=1\,.
$$

Then, given an initial probability density $f_{\{c\}}^{in}\equiv f_{\{c\}}^{in}(x,v)$ on the single-particle 
phase-space with support outside the system of obstacles defined by $\{c\}$, we define its evolution 
under the billiard flow by the formula
$$
f(t,x,v,\{c\} )=f^{in}_{\{c\}}(X(-t;x,v,\{c\} ),V(-t;x,v,\{c\}))\,,\quad t\ge 0\,.
$$

Let $\tau_1(x,v,\{c\} ),\tau_2(x,v,\{c\} ),\ldots,\tau_j(x,v,\{c\} ),\ldots$ be the  sequence of collision times 
for a particle starting from $x$ in the direction $-v$ at $t=0$ in the configuration of obstacles $\{c\}$: in 
other words, 
$$
\begin{aligned}
{}&\tau_j(x,v,\{c\} )=
\\
&\quad\sup\{t\,|\,\#\{s\in[0,t]\,|\,\hbox{dist}(X(-s,x,v,\{c\} );\{c\} )=r\}=j-1\}\,.
\end{aligned}
$$

Denoting $\tau_0=0$ and $\Dlt\tau_k=\tau_k-\tau_{k-1}$, the evolved single-particle density $f$ is 
a.e. defined by the formula 
$$
\begin{aligned}
{}&f(t,x,v,\{c\} )=f^{in}(x-tv,v)\indc_{t<\tau_1}
\\
&+\!\sum_{j\ge 1}\!f^{in}\!\left(x\!-\!\!\!\sum_{k=1}^j\!\!
	\Dlt\tau_kV(-\tau_k^-)\!-\!(t-\tau_j)V(-\tau_j^+),
		V(-\tau_j^+)\right)\!\indc_{\tau_j<t<\tau_{j+1}}\,.
\end{aligned}
$$

In the case of physically admissible initial data, there should be no particle located inside an 
obstacle. Hence we assumed that $f^{in}_{\{c\}}=0$ in the union of all the disks of radius $r$ centered 
at the $c_j\in\{c\}$. By construction, this condition is obviously preserved by the billiard flow, so that 
$f(t,x,v,\{c\})$ also vanishes whenever $x$ belongs to a disk of radius $r$ centered at any $c_j\in\{c\}$.

As we shall see shortly, when dealing with bounded initial data, this constraint disappears in the (yet 
undefined) Boltzmann-Grad limit, as the volume fraction occupied by the obstacles vanishes in that 
limit. 

Therefore, we shall henceforth neglect this difficulty and proceed as if $f^{in}$ were any bounded 
probability density on $\bR^2\times\bS^1$.

Our goal is to average the summation above in the obstacle configuration $\{c\} $ under the Poisson
distribution, and to identify a scaling on the obstacle radius $r$ and the parameter $n$ of the Poisson 
distribution leading to a nontrivial limit.

The parameter $n$ has the following important physical interpretation. The expected number of 
obstacle centers to be found in any measurable subset $\Om$ of the Euclidian plane $\bR^2$ is
$$
\sum_{p\ge 0}p\hbox{Prob}(\{\{c\} \,|\,\#(\Om\cap\{c\} )=p\})
	=\sum_{p\ge 0}pe^{-n|\Om|}\frac{(n|\Om|)^p}{p!}=n|\Om|
$$
so that
$$
n=\#\hbox{ obstacles per unit surface in $\bR^2$.}
$$

The average of the first term in the summation defining $f(t,x,v,\{c\} )$ is
$$
f^{in}(x-tv,v)\la\indc_{t<\tau_1}\ra=f^{in}(x-tv,v)e^{-n2rt}
$$
(where $\la\,\cdot\,\ra$ denotes the mathematical expectation) since the condition $t<\tau_1$ means 
that the tube of width $2r$ and length $t$ contains no obstacle center.

\begin{figure}

\centerline{\includegraphics[width=7.0cm]{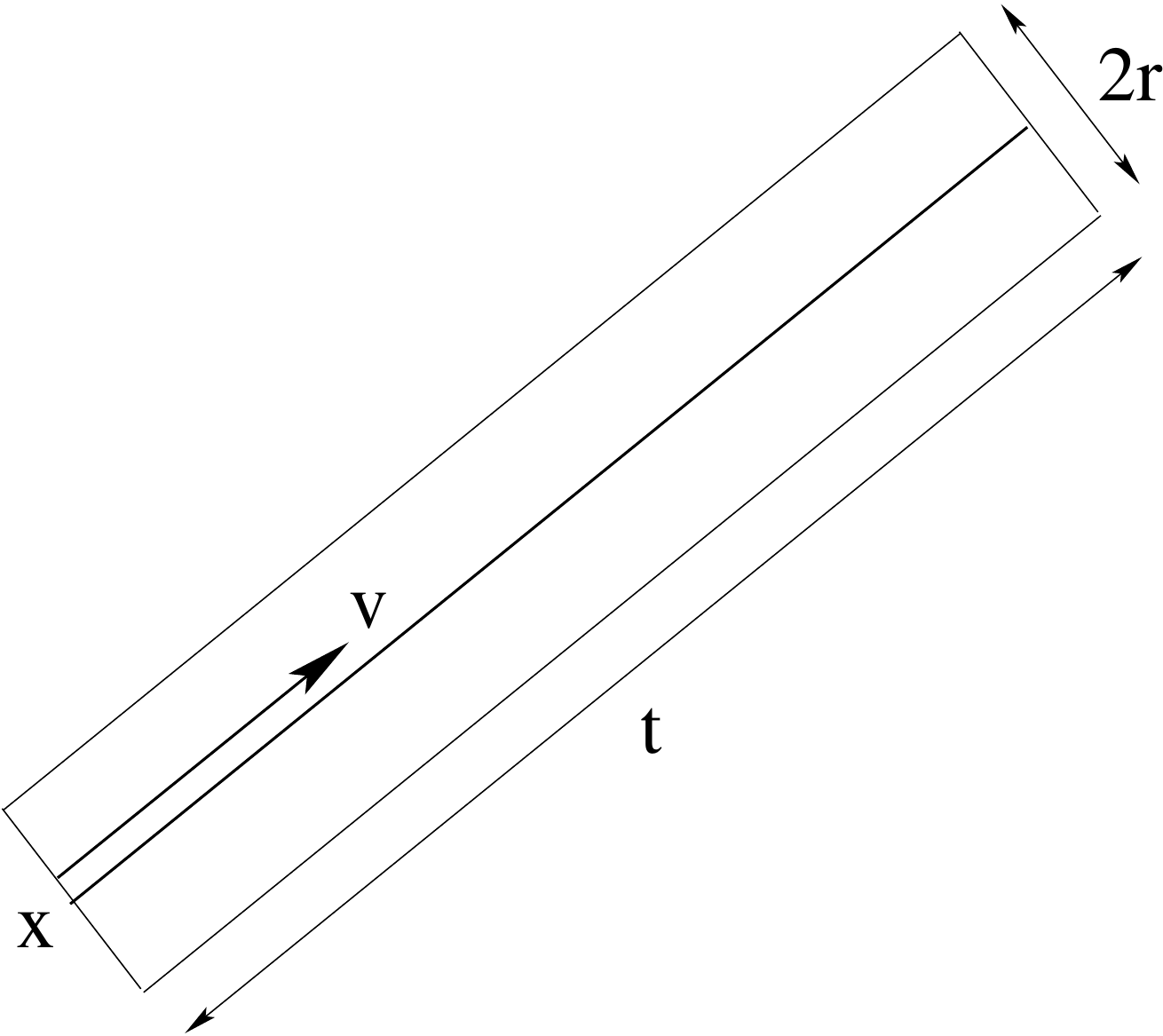}}

\caption{The tube corresponding with the first term in the series
expansion giving the particle density}

\end{figure}

Henceforth, we seek a scaling limit corresponding to small obstacles, i.e. $r\to 0$ and a large number 
of obstacles per unit volume, i.e. $n\to\infty$.

There are obviously many possible scalings satisfying this requirement. Among all these scalings, the 
Boltzmann-Grad scaling in space dimension $2$ is defined by the requirement that the average 
over obstacle configurations of the first term in the series expansion for the particle density $f$ has a 
nontrivial limit.


\bigskip
\noindent
\fbox{\sc Boltzmann-Grad scaling in space dimension 2} 

\smallskip
In order for the average of the first term above to have a nontrivial limit, one must have
$$
r\to 0^+\hbox{ and }n\to+\infty\hbox{ in such a way that }2nr\to\si>0\,.
$$

\bigskip
Under this assumption
$$
\la f^{in}(x-tv,v)\indc_{t<\tau_1}\ra\to f^{in}(x-tv,v)e^{-\si t}\,.
$$

Gallavotti's idea is that this first term corresponds with the solution at time $t$ of the equation
$$
\begin{aligned}
(\d_t+v\cdot\grad_x)f&=-nrf\int_{|\om|=1\atop\om\cdot v>0}\cos(v,\om)d\om=-2nrf
\\
f\rstr_{t=0}&=f^{in}
\end{aligned}
$$
that involves only the loss part in the Lorentz collision integral, and that the (average over obstacle 
configuration of the) subsequent terms in the sum defining the particle density $f$ should converge 
to the Duhamel formula for the Lorentz kinetic equation.

After this necessary preliminaries, we can state Gallavotti's theorem.

\begin{Thm}[Gallavotti \cite{Gallavotti1972}]
Let $f^{in}$ be a continuous, bounded probability density on $\bR^2\times\bS^1$, and let 
$$
f_r(t,x,v,\{c\} )=f^{in}((X^r,V^r)(-t,x,v,\{c\}))\,,
$$
where $(t,x,v)\mapsto (X^r,V^r)(t,x,v,\{c\} )$ is the billiard flow in the system of disks of radius $r$ 
centered at the elements of $\{c\}$. Assuming that the obstacle centers are distributed under the 
Poisson law of parameter $n=\si/2r$ with $\si>0$, the expected single particle density 
$$
\la f_r(t,x,v,\cdot)\ra\to f(t,x,v)\hbox{ in }L^1(\bR^2\times\bS^1)
$$
uniformly on compact $t$-sets, where $f$ is the solution of the Lorentz kinetic equation
$$
\begin{aligned}
(\d_t+v\cdot\grad_x)f+\si f
	&=\si\int_0^{2\pi} f(t,x,R[\b]v)\sin\tfrac{\b}2 \tfrac{d\b}4\,,
\\
f\rstr_{t=0}&=f^{in}\,,
\end{aligned}
$$
where $R[\b]$ denotes the rotation of an angle $\b$.
\end{Thm}

\begin{proof}[End of the proof of Gallavotti's theorem]
The general term in the summation giving $f(t,x,v,\{c\} )$ is
$$
f^{in}\!\left(x\!-\!\!\!\sum_{k=1}^j\!\!\Dlt\tau_kV^r(-\tau_k^-)\!-\!(t-\tau_j)V^r(-\tau_j^+),
	V^r(-\tau_j^+)\right)\!\indc_{\tau_j<t<\tau_{j+1}}\,,
$$
and its average under the Poisson distribution on $\{c\} $ is
$$
\begin{aligned}
\int 
f^{in}\!\left(x-\!\sum_{k=1}^j\Dlt\tau_kV^r(-\tau_k^-)-(t-\tau_j)V^r(-\tau_j^+),
	V^r(-\tau_j^-)\right)
\\
e^{-n|T(t;c_1,\ldots,c_j)|}\frac{n^jdc_1\ldots dc_j}{j!}\,,
\end{aligned}
$$
where $T(t;c_1,\ldots,c_j)$ is the tube of width $2r$ around the particle trajectory colliding first with 
the obstacle centered at $c_1$, \dots, and whose $j$-th collision is with the obstacle centered at 
$c_j$. 

As before, the surface of that tube is
$$
|T(t;c_1,\ldots,c_j)|=2rt+O(r^2)\,.
$$

\begin{figure}

\centerline{\includegraphics[width=8.0cm]{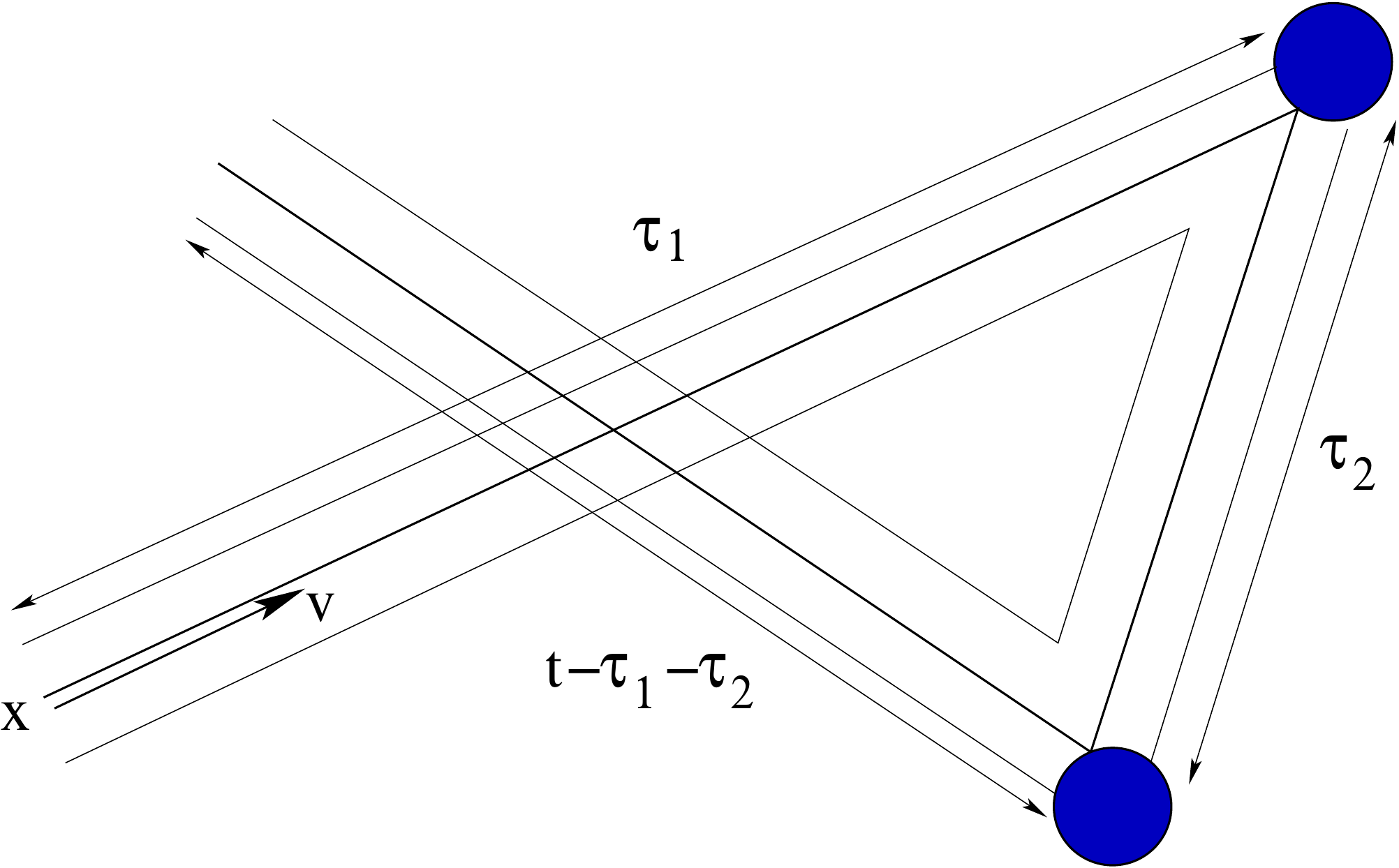}}

\caption{The tube $T(t,c_1,c_2)$ corresponding with the third term in the series
expansion giving the particle density}

\end{figure}

In the $j$-th term, change variables by expressing the positions of the $j$ encountered obstacles in 
terms of free flight times and deflection angles:  
$$
(c_1,\ldots,c_j)\mapsto(\tau_1,\ldots,\tau_j;\b_1,\ldots,\b_j)\,.
$$
The volume element in the $j$-th integral  is changed into
$$
\tfrac{dc_1\ldots dc_j}{j!}=r^j\sin\tfrac{\b_1}{2}\ldots\sin\tfrac{\b_j}{2}\,\,
	\tfrac{d\b_1}2\ldots\tfrac{d\b_j}2d\tau_1\ldots d\tau_j\,. 
$$
The measure in the left-hand side is invariant by permutations of $c_1,\ldots,c_j$; on the right-hand 
side, we assume that
$$
\tau_1<\tau_2<\ldots<\tau_j\,,
$$
which explains why $1/j!$ factor disappears in the right-hand side.

\begin{figure}

\centerline{\includegraphics[width=9.0cm]{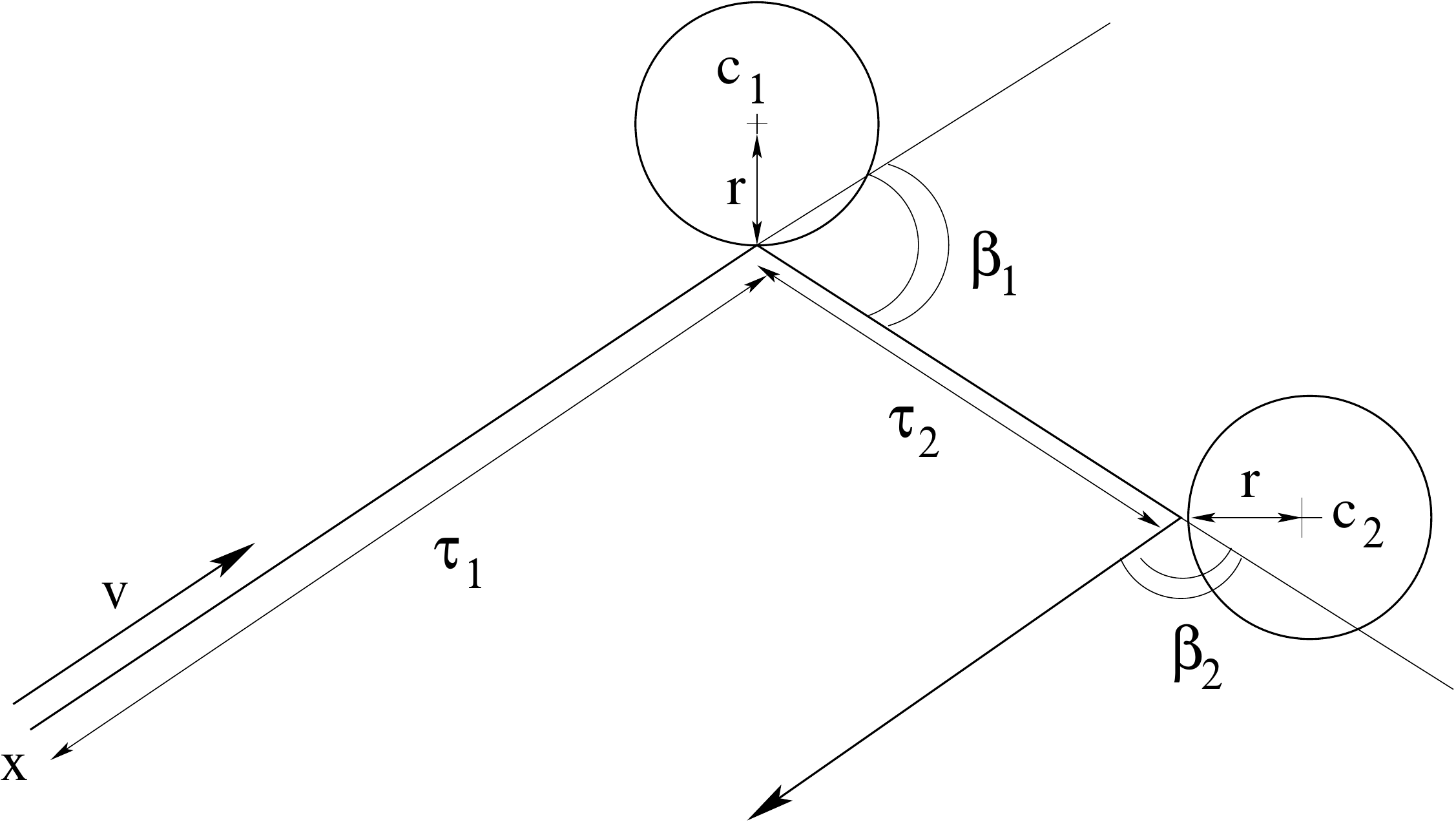}}

\caption{The substitution $(c_1,c_2)\mapsto(\tau_1,\tau_2,\b_1,\b_2)$}

\end{figure}

The substitution above is one-to-one only if the particle does not hit twice the same obstacle. Define 
therefore
$$
\begin{aligned}
{}&A_r(T,x,v)=\{\{c\}\,|\,\hbox{ there exists }0<t_1<t_2<T\hbox{ and }j\in\bN\hbox{ s.t. }
\\
&\qquad\qquad\qquad\hbox{dist}(X^r(t_1,x,v,\{c\}),c_j)=\hbox{dist}(X^r(t_2,x,v,\{c\}),c_j)=r\}
\\
&\qquad=\bigcup_{j\ge 1}\{\{c\}\,|\,
	\hbox{dist}(X^r(t,x,v,\{c\}),c_j)=r\hbox{ for some }0<t_1<t_2<T\}\,,
\end{aligned}
$$
and set
$$
\begin{aligned}
f^M_r(t,x,v,\{c\})&=f_r(t,x,v,\{c\})-f^R_r(t,x,v,\{c\})\,,
\\
f^R_r(t,x,v,\{c\})&=f_r(t,x,v,\{c\})\indc_{A_r(T,x,v)}(\{c\})\,,
\end{aligned}
$$
respectively the Markovian part and the recollision part in $f_r$.

After averaging over the obstacle configuration $\{c\} $, the contribution of the $j$-th term in $f^M_r$ is, 
to leading order in $r$:
$$
\begin{aligned}
(2nr)^je^{-2nrt}\int_{0<\tau_{1}<\ldots<\tau_j<t}\int_{[0,{2\pi}]^j}
\sin\tfrac{\b_1}{2}\ldots\sin\tfrac{\b_j}{2}\tfrac{d\b_1}4\ldots\tfrac{d\b_j}4d\tau_1\ldots d\tau_j 
\\
\times f^{in}\left(x\!-\!\!\!\sum_{k=1}^j\!\!\Dlt\tau_k
R\left[\sum_{l=1}^{k-1}\b_l\right]v\!-\!(t-\tau_j)R\left[\sum_{l=1}^{j-1}\b_l\right]v,
	R\left[\sum_{l=1}^{j}\b_l\right]v\right)\,.
\end{aligned}
$$
It is dominated by
$$
\|f^{in}\|_{L^\infty}O(\si)^je^{-O(\si)t}\frac{t^j}{j!}
$$
which is the general term of a converging series.

Passing to the limit as $n\to+\infty$, $r\to 0$ so that $2rn\to\si$, one finds (by dominated convergence 
in the series) that
$$
\begin{aligned}
\la f^M_r(t,x,v,\{c\} )\ra\to e^{-\si t}f^{in}(x-tv,v)
\\
+
\si e^{-\si t}\int_0^t\int_0^{2\pi}
f^{in}(x-\tau_1v-(t-\tau_1)R[\b_1]v,R[\b_1]v)\sin\tfrac{\b_1}{2}\tfrac{d\b_1}4 d\tau_1
\\
+\sum_{j\ge 2}\si^je^{-\si t}\int_{0<\tau_{j}<\ldots<\tau_1<t}\int_{[0,{2\pi}]^j}
	\sin\tfrac{\b_1}{2}\ldots\sin\tfrac{\b_j}{2}
\\
\times f^{in}\left(x\!-\!\!\!\sum_{k=1}^j\!\!\Dlt\tau_k
R\left[\sum_{l=1}^{k-1}\b_l\right]v\!-\!(t-\tau_j)R\left[\sum_{l=1}^{j-1}\b_l\right]v,
	R\left[\sum_{l=1}^{j}\b_l\right]v\right)
\\
\times\tfrac{d\b_1}4\ldots\tfrac{d\b_j}4d\tau_1\ldots d\tau_j\,,
\end{aligned}
$$
which is the Duhamel series giving the solution of the Lorentz kinetic equation.

Hence, we have proved that
$$
\la f^M_r(t,x,v,\cdot)\ra\to f(t,x,v)\hbox{ uniformly on bounded sets as }r\to 0^+\,,
$$
where $f$ is the solution of the Lorentz kinetic equation. One can check by a straightforward 
computation that the Lorentz collision integral satisfies the property
$$
\int_{\bS^1}\cC(\phi)(v)dv=0\hbox{ for each }\phi\in L^\infty(\bS^1)\,.
$$
Integrating both sides of the Lorentz kinetic equation in the variables $(t,x,v)$ over 
$[0,t]\times\bR^2\times\bS^1$ shows that the solution $f$ of that equation satisfies
$$
\iint_{\bR^2\times\bS^1}f(t,x,v)dxdv=\iint_{\bR^2\times\bS^1}f^{in}(x,v)dxdv
$$
for each $t>0$.

On the other hand, the billiard flow $(X,V)(t,\cdot,\cdot,\{c\})$ obviously leaves the uniform measure 
$dxdv$ on $\bR^2\times\bS^1$ (i.e. the particle number) invariant, so that, for each $t>0$ and each 
$r>0$, 
$$
\iint_{\bR^2\times\bS^1}f_r(t,x,v,\{c\})dxdv=\iint_{\bR^2\times\bS^1}f^{in}(x,v)dxdv\,.
$$
We therefore deduce from Fatou's lemma that
$$
\begin{aligned}
\la f^R_r\ra\to 0\hbox{ in }L^1(\bR^2\times\bS^1)\hbox{ uniformly on bounded $t$-sets}
\\
\la f^M_r\ra\to f\hbox{ in }L^1(\bR^2\times\bS^1)\hbox{ uniformly on bounded $t$-sets}
\end{aligned}
$$
which concludes our sketch of the proof of Gallavotti's theorem. 
\end{proof}

For a complete proof, we refer the interested reader to \cite{Gallavotti1972, Gallavotti1999}.

Some remarks are in order before leaving Gallavotti's setting for the Lorentz gas with the Poisson 
distribution of obstacles.

Assuming no external force field as done everywhere in the present paper is not as inocuous as it may 
seem. For instance, in the case of Poisson distributed holes --- i.e. purely absorbing obstacles, so that 
particles falling into the holes disappear from the system forever --- the presence of an external force 
may introduce memory effects in the Boltzmann-Grad limit, as observed by L. Desvillettes and V. Ricci 
\cite{DesvRicci2004}.

Another remark is about the method of proof itself. One has obtained the Lorentz kinetic equation
\textit{after} having obtained an explicit formula for the solution of that equation. In other words, the
equation is deduced from the solution --- which is a somewhat unusual situation in mathematics.
However, the same is true of Lanford's derivation of the Boltzmann equation \cite{Lanford1975},
as well as of the derivation of several other models in nonequilibrium statistical mechanics. For an
interesting comment on this issue, see \cite{CerciIllnerPulvi1994}, on p. 75.


\section{Santal\'o's formula\\ for the geometric mean free path}


From now on, we shall abandon the random case and concentrate our efforts on the periodic Lorentz 
gas. 

Our first task is to define the Boltzmann-Grad scaling for periodic systems of spherical obstacles. In the 
Poisson case defined above, things were relatively easy: in space dimension $2$, the Boltzmann-Grad 
scaling was defined by the prescription that the number of obstacles per unit volume tends to infinity 
while the obstacle radius tends to $0$ in such a way that
$$
\#\hbox{ obstacles per unit volume }\times\hbox{ obstacle radius }\to\si>0\,.
$$

The product above has an interesting geometric meaning even without assuming a Poisson distribution 
for the obstacle centers, which we shall briefly discuss before going further in our analysis of the periodic 
Lorentz gas.

Perhaps the most important scaling parameter in all kinetic models is the mean free path. This is by no 
means a trivial notion, as will be seen below. As suggested by the name itself, any notion of mean free 
path must involve first the notion of free path length, and then some appropriate probability measure
under which the free path length is averaged.

For simplicity, the only periodic distribution of obstacles considered below is the set of balls of radius 
$r$ centered at the vertices of a unit cubic lattice in the $D$-dimensional Euclidian space.

Correspondingly, for each $r\in(0,\tfrac12)$, we define the domain left free for particle motion, also 
called the ``billiard table" as
$$
Z_r=\{x\in\bR^D\,|\,\hbox{dist}(x,\bZ^D)>r\}\,.
$$

\begin{figure}

\includegraphics[width=10.0cm]{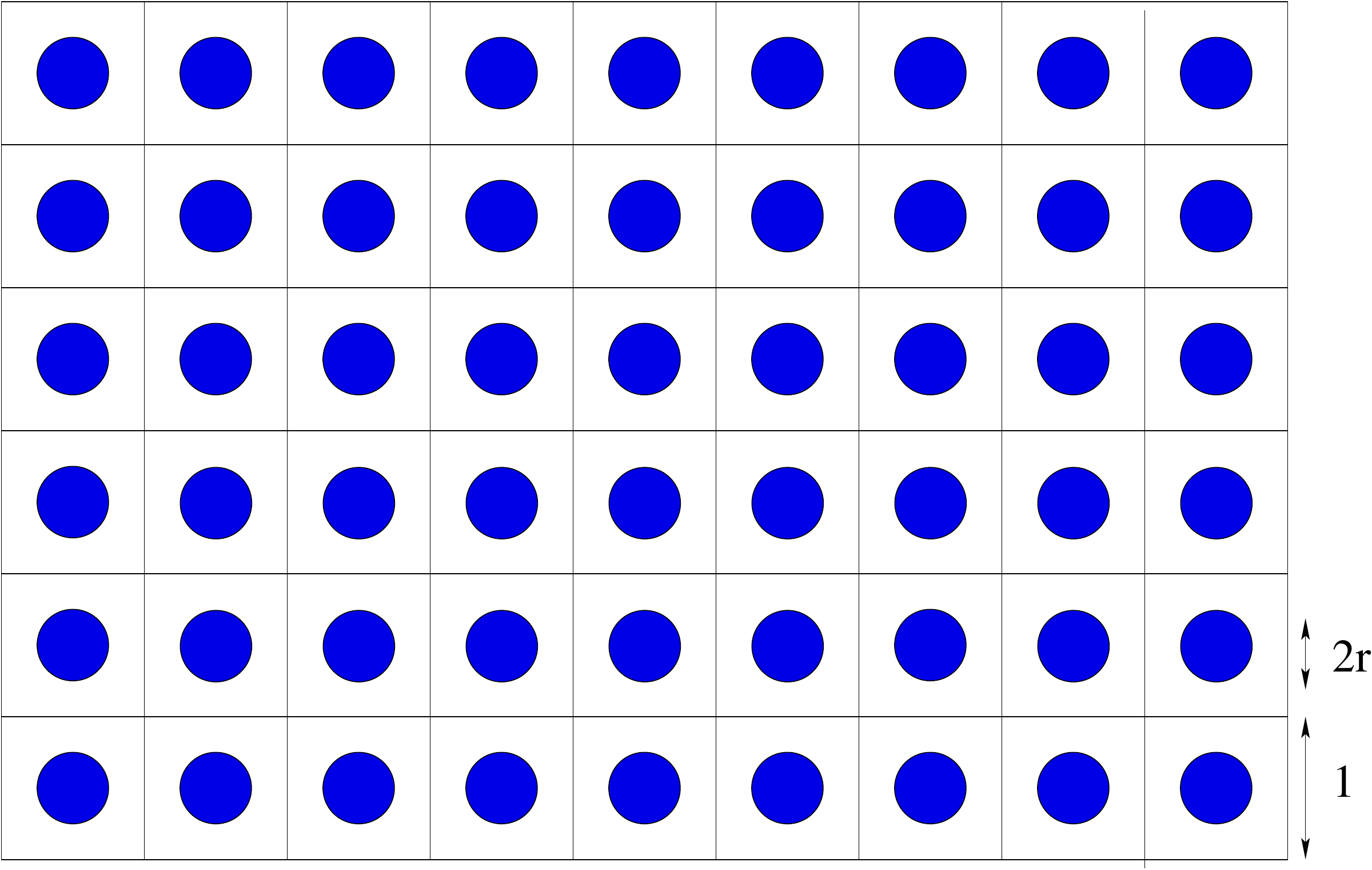}

\caption{The periodic billiard table}

\end{figure}

Defining the free path length in the billiard table $Z_r$ is easy:  the free path length starting from 
$x\in Z_r$ in the direction $v\in\bS^{D-1}$ is
$$
\tau_r(x,v)=\min\{t>0\,|\,x+tv\in\d Z_r\}\,.
$$

Obviously, for each $v\in\bS^{D-1}$ the free path length $\tau_r(\cdot,v)$ in the direction $v$ can be 
extended continuously to 
$$
\{x\in\d Z_r\,|\,v\cdot n_x\not=0\}\,,
$$
where $n_x$ denotes the unit normal vector to $\d Z_r$ at the point $x\in\d Z_r$ pointing towards 
$Z_r$.

With this definition, the mean free path is the quantity defined as
$$
\hbox{Mean Free Path}=\la\tau_r\ra\,,
$$
where the notation $\la\cdot\ra$ designates the average under some appropriate probability measure 
on $\overline{Z_r}\times\bS^{D-1}$. 

\begin{figure}

\includegraphics[width=9.0cm]{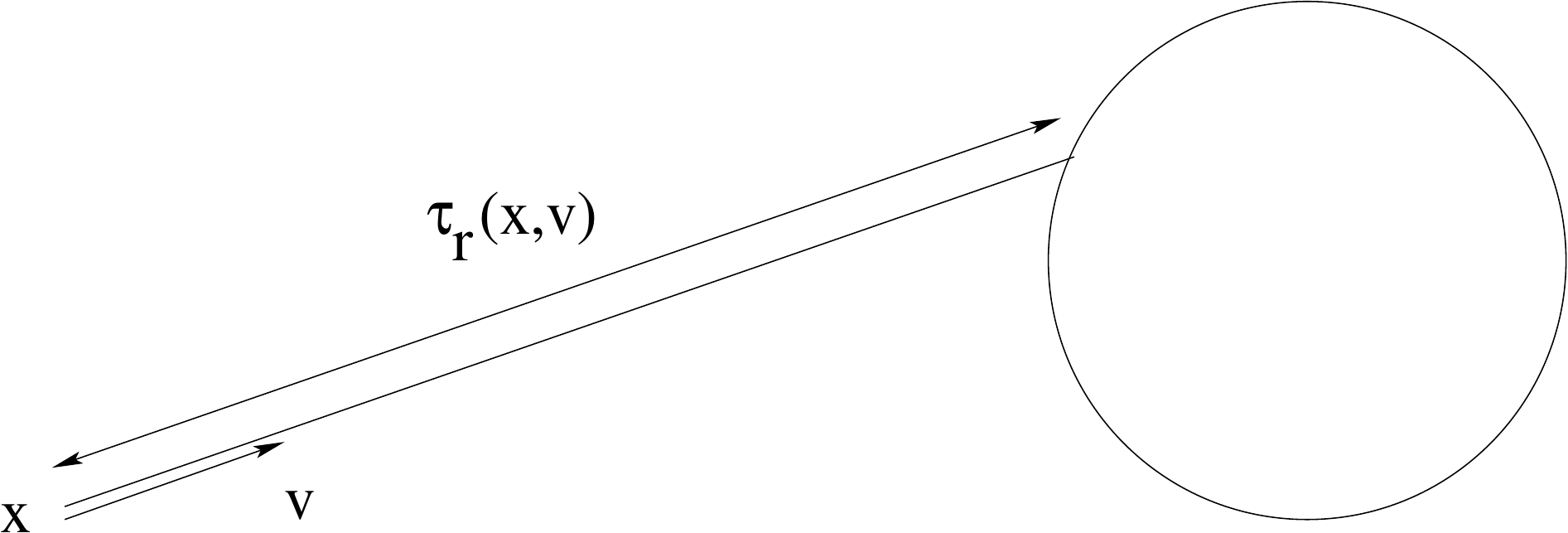}

\caption{The free path length}

\end{figure}

A first ambiguity in the notion of mean free path comes from the fact that there are two fairly natural 
probability measures for the Lorentz gas.

The first one is the uniform probability measure on $Z_r/\bZ^D\times\bS^{D-1}$
$$
d\mu_r(x,v)=\frac{dxdv}{|Z_r/\bZ^D|\,|\bS^{D-1}|}
$$
that is invariant under the billiard flow --- the notation $|\bS^{D-1}|$ designates the $D-1$-dimensional 
uniform measure of the unit sphere $\bS^{D-1}$. This measure is obviously invariant under the billiard 
flow 
$$
(X_r,V_r)(t,\cdot,\cdot):\,Z_r\times\bS^{D-1}\to Z_r\times\bS^{D-1}
$$
defined by
$$
\left\{\begin{matrix}\dot{X}_r=V_r
	\\ \dot{V}_r=0\end{matrix}\right.\quad\hbox{ whenever }X(t)\notin\d Z_r
$$
while
$$
\left\{\begin{array}l
X_r(t^+)=X_r(t^-)=:X_r(t)\hbox{ if }X(t^\pm)\in\d Z_r\,,
\\ 
V_r(t^+)=\cR[n_{X_r(t)}]V_r(t^-)
\end{array}\right.
$$
with $\cR[n]v=v-2v\!\cdot\!nn$ denoting the reflection with respect to the hyperplane $(\bR n)^\bot$.

The second such probability measure is the invariant measure of the billiard map 
$$
d\nu_r(x,v)=
	\frac{v\!\cdot\!n_xdS(x)dv}{v\!\cdot\!n_xdxdv\hbox{-meas}(\Gamma^r_+/\bZ^D)}
$$
where $n_x$ is the unit inward normal at $x\in\d Z_r$, while $dS(x)$ is the $D-1$-dimensional surface 
element on $\d Z_r$, and
$$
\Gamma^r_+:=\{(x,v)\in\d Z_r\times\bS^{D-1}\,|\,v\cdot n_x>0\}\,.
$$
The billiard map $\cB_r$ is the map 
$$
\Gamma^r_+\ni(x,v)\mapsto\cB_r(x,v):=(X_r,V_r)(\tau_r(x,v);x,v)\in\Gamma^r_+\,,
$$
which obviously passes to the quotient modulo $\bZ^D$-translations:
$$
\cB_r:\;\Gamma^r_+/\bZ^D\to\Gamma^r_+/\bZ^D\,.
$$
In other words, given the position $x$ and the velocity $v$ of a particle immediatly after its first collision 
with an obstacle, the sequence $(\cB^n_r(x,v))_{n\ge 0}$ is the sequence of all collision points and 
post-collision velocities on that particle's trajectory.

With the material above, we can define a first, very natural notion of mean free path, by setting
$$
\hbox{Mean Free Path}
	=\lim_{N\to+\infty}\frac1{N}\sum_{k=0}^{N-1}\tau_r(\cB^k_r(x,v))\,.
$$
Notice that, for $\nu_r$-a.e. $(x,v)\in\Gamma^+_r/\bZ^D$, the right hand side of the equality above is 
well-defined by the Birkhoff ergodic theorem. If the billiard map $\cB_r$ is ergodic for the measure 
$\nu_r$, one has
$$
\lim_{N\to+\infty}\frac1{N}\sum_{k=0}^{N-1}\tau_r(\cB^k_r(x,v))
	=\int_{\Gamma^r_+/{\bZ^D}}\tau_rd\nu_r\,,
$$
for $\nu_r$-a.e. $(x,v)\in\Gamma^r_+/\bZ^D$.

Now, a very general formula for computing the right-hand side of the above equality was found by 
the great spanish mathematician L. A. Santal\'o in 1942. In fact, Santal\'o's argument applies to 
situations that are considerably more general, involving for instance curved trajectories instead of 
straight line segments, or obstacle distributions other than periodic. The reader interested in these
questions is referred to Santal\'o's original article \cite{Santalo1943}.

\begin{figure}

\includegraphics[width=6.0cm]{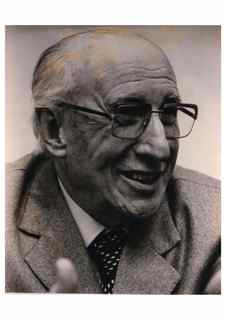}

\caption{Luis Antonio Santal\'o Sors (1911-2001)}

\end{figure}

Here is

\bigskip
\noindent
\fbox{\sc Santal\'o's formula for the geometric mean free path}

\smallskip
One finds that
$$
\ell_r:=\int_{\Gamma^r_+/{\bZ^D}}\tau_r(x,v)d\nu_r(x,v)
	=\frac{1-|\bB^D|r^D}{|\bB^{D-1}|r^{D-1}}
$$
where $\bB^D$ is the unit ball of $\bR^D$ and $|\bB^D|$ its $D$-dimensional
Lebesgue measure.

\smallskip
In fact, one has the following slightly more general

\begin{Lem} [H.S. Dumas, L. Dumas, F. Golse \cite{Dumas2Golse1996}] For $f\in C^1(\bR_+)$ 
such that $f(0)=0$, one has
$$
\iint_{\Gamma^r_+/\bZ^D}f(\tau_r(x,v))v\cdot n_xdS(x)dv
		=\iint_{(Z_r/\bZ^D)\times\bS^{D-1}}f'(\tau_r(x,v))dxdv\,.
$$
\end{Lem}

Santal\'o's formula is obtained by setting $f(z)=z$ in the identity above, and expressing both 
integrals in terms of the normalized measures $\nu_r$ and $\mu_r$.

\begin{proof} 
For each $(x,v)\in Z_r\times\bS^{D-1}$ one has
$$
\tau_r(x+tv,v)=\tau_r(x,v)-t\,,
$$
so that
$$
\frac{d}{dt}\tau_r(x+tv,v)=-1\,.
$$
Hence $\tau_r(x,v)$ solves the transport equation
$$
\left\{\begin{array}{lll}
v\cdot\grad_x\tau_r(x,v)=-1\,,&\quad x\in Z_r\,,\,\,&v\in\bS^{D-1}\,,
\\
\tau_r(x,v)=0\,,&\quad x\in\d Z_r\,,\,\,&v\cdot n_x<0\,.
\end{array}\right.
$$
Since $f\in C^1(\bR_+)$ and $f(0)=0$, one has
$$
\left\{\begin{array}{lll}
v\cdot\grad_xf(\tau_r(x,v))=-f'(\tau_r(x,v))\,,&\quad x\in Z_r\,,\,\,&v\in\bS^{D-1}\,,
\\
f(\tau_r(x,v))=0\,,&\quad x\in\d Z_r\,,\,\,&v\cdot n_x<0\,.
\end{array}\right.
$$
Integrating both sides of the equality above, and applying Green's formula shows that
$$
\begin{aligned}
-\iint_{(Z_r/\bZ^D)\times\bS^{D-1}}&f'(\tau_r(x,v))dxdv
\\
&=
\iint_{(Z_r/\bZ^D)\times\bS^{D-1}}v\cdot\grad_x(f(\tau_r(x,v)))dxdv
\\
&=-\iint_{(\d Z_r/\bZ^D)\times\bS^{D-1}}f(\tau_r(x,v))v\cdot n_xdS(x)dv
\end{aligned}
$$
--- beware the unusual sign in the right-hand side of the second equality above, coming from the 
orientation of the unit normal $n_x$, which is pointing towards $Z_r$. 
\end{proof}

\smallskip
With the help of Santal\'o's formula, we define the Boltzmann-Grad limit for the Lorentz gas with 
periodic as well as random distribution of obstacles as follows:

\bigskip
\noindent
\fbox{\sc Boltzmann-Grad scaling}

\smallskip
The Boltzmann-Grad scaling for the periodic Lorentz gas in space dimension $D$ corresponds with 
the following choice of parameters:
$$
\begin{aligned}
\hbox{distance between neighboring lattice points}&=\eps\ll 1\,,
\\
\hbox{obstacle radius}&=r\ll 1\,,
\\
\hbox{mean free path}&=\ell_r\to\frac1\si>0\,.
\end{aligned}
$$ 

Santal\'o's formula indicates that one should have
$$
r\sim c\eps^{\frac{D}{D-1}}
	\hbox{ with }c=\left(\frac{\si}{|\bB^{D-1}|}\right)^{-\frac1{D-1}}\hbox{ as }\eps\to 0^+\,.
$$

Therefore, given an initial particle density $f^{in}\in C_c(\bR^D\times\bS^{D-1})$, we define $f_r$ to 
be
$$
f_r(t,x,v)=f^{in}\left(r^{D-1}X_r\left(-\frac{t}{r^{D-1}};\frac{x}{r^{D-1}},v\right),
	V_r\left(-\frac{t}{r^{D-1}};\frac{x}{r^{D-1}},v\right)\right)
$$
where $(X_r,V_r)$ is the billiard flow in $Z_r$ with specular reflection on $\d Z_r$.

Notice that this formula defines $f_r$ for $x\in Z_r$ only, as the particle density should remain $0$ for 
all time in the spatial domain occupied by the obstacles. As explained in the previous section, this is 
a set whose measure vanishes in the Boltzmann-Grad limit, and we shall always implicitly extend the 
function $f_r$ defined above by $0$ for $x\notin Z_r$. 

Since $f^{in}$ is a bounded function on $Z_r\times\bS^{D-1}$, the family $f_r$ defined above is a 
bounded family of $L^\infty(\bR^D\times\bS^{D-1})$. By the Banach-Alaoglu theorem, this family is 
therefore relatively compact for the weak-* topology of $L^\infty(\bR_+\times\bR^D\times\bS^{D-1})$.

\smallskip
\textbf{Problem:} to find an equation governing the $L^\infty$ weak-* limit points of the scaled number 
density $f_r$ as $r\to 0^+$.

\smallskip
In the sequel, we shall describe the answer to this question in the $2$-dimensional case ($D=2$.)


\section{Estimates for the distribution of free-path lengths}


In the proof of Gallavotti's theorem for the case of a Poisson distribution of obstacles in space 
dimension $D=2$, the probability that a strip of width $2r$ and length $t$ does not meet any obstacle 
is $e^{-2nrt}$, where $n$ is the parameter of the Poisson distribution --- i.e. the average number of 
obstacles per unit surface. 

This accounts for the loss term
$$
f^{in}(x-tv,v)e^{-\si t}
$$
in the Duhamel series for the solution of the Lorentz kinetic equation, or of the term $-\si f$ on the 
right-hand side of that equation written in the form
$$
(\d_t+v\cdot\grad_x)f=-\si f+\si\int_0^{2\pi}f(t,x,R(\b)v)\sin\tfrac{\b}2\tfrac{d\b}{4}\,.
$$

Things are fundamentally different in the periodic case. To begin with, there are infinite strips included 
in the billiard table $Z_r$ which \textit{never} meet any obstacle.

\begin{figure}

\includegraphics[width=8.0cm]{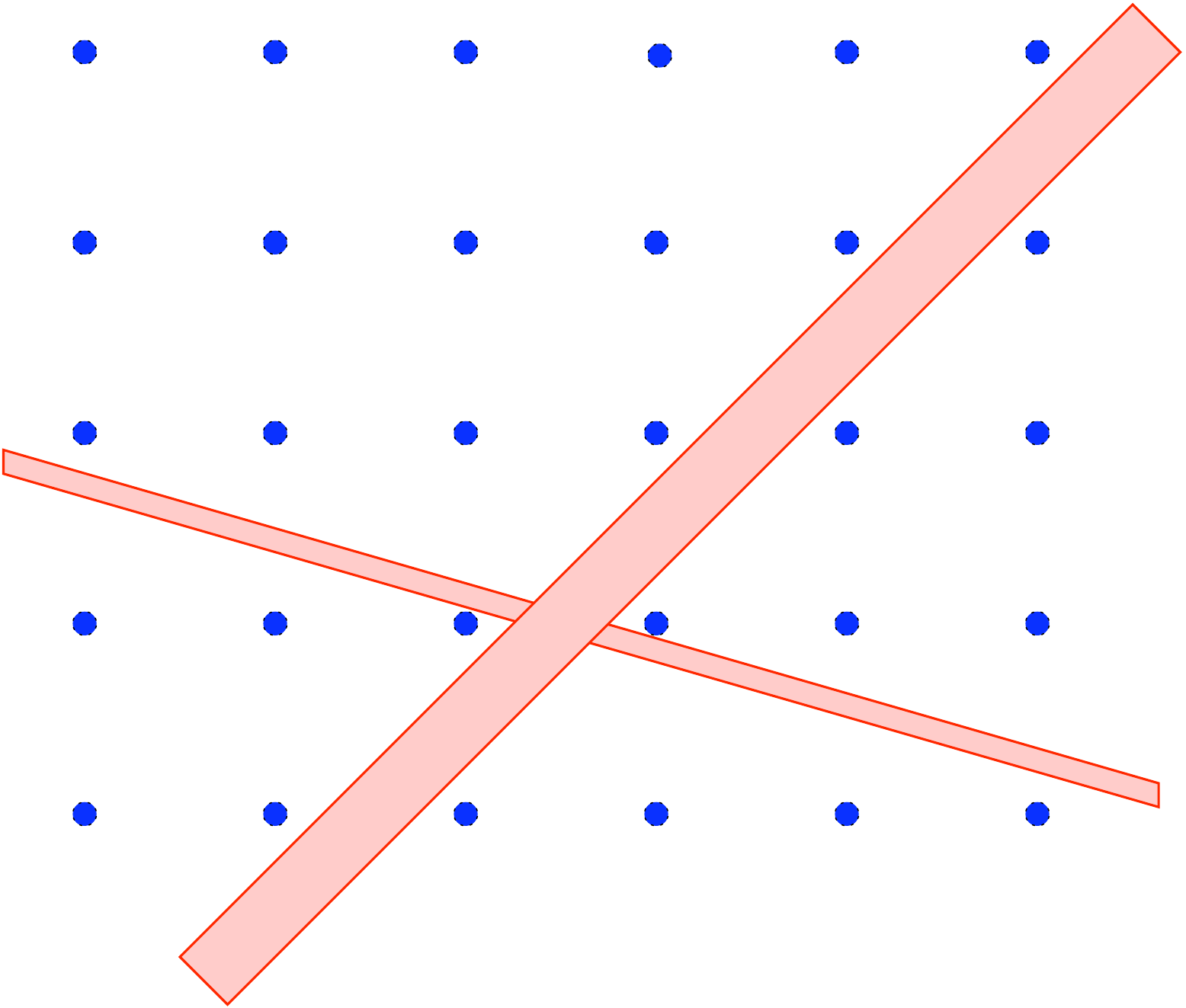}

\caption{Open strips in the periodic billiard table that never meet any obstacle}

\end{figure}

The contribution of the 1-particle density leading to the loss term in the Lorentz kinetic equation is, in 
the notation of the proof of Gallavotti's theorem
$$
f^{in}(x-tv,v)\indc_{t<\tau_1(x,v,\{c\})}\,.
$$
The analogous term in the periodic case is
$$
f^{in}(x-tv,v)\indc_{t<r^{D-1}\tau_r(x/r,-v)}
$$
where $\tau_r(x,v)$ is the free-path length in the periodic billiard table $Z_r$ starting from $x\in Z_r$ in 
the direction $v\in\bS^1$.

Passing to the $L^\infty$ weak-* limit as $r\to 0$ reduces to finding
$$
\lim_{r\to 0}\indc_{t<r^{D-1}\tau_r(x/r,-v)}
	\hbox{ in }w^*-L^\infty(\bR^2\times\bS^1)\,
$$
--- possibly after extracting a subsequence $r_n\downarrow 0$. As we shall see below, this involves 
the distribution of $\tau_r$ under the probability measure $\mu_r$ introduced in the discussion of 
Santal\'o's formula --- i.e. assuming the initial position $x$ and direction $v$ to be independent and
uniformly distributed on $(\bR^D/\bZ^D)\times\bS^{D-1}$.

We define the (scaled) distribution under $\mu_r$ of free path lengths $\tau_r$ to be
$$
\Phi_r(t):=\mu_r(\{(x,v)\in(Z_r/\bZ^D)\times\bS^{D-1}\,|\,\tau_r(x,v)>t/r^{D-1}\})\,.
$$

Notice the scaling $t\mapsto t/r^{D-1}$ in this definition. In space dimension $D$, Santal\'o's formula 
shows that
$$
\iint_{\Gamma_r^+/\bZ^D}\tau_r(x,v)d\nu_r(x,v)\sim\tfrac{1}{|\bB^{D-1}|}r^{1-D}\,,
$$
and this suggests that the free path length $\tau_r$ is a quantity of the order of $1/r^{D-1}$. (In fact, 
this argument is not entirely convincing, as we shall see below.)

In any case, with this definition of the distribution of free path lengths under $\mu_r$, one arrives at 
the following estimate.

\begin{Thm}[Bourgain-Golse-Wennberg \cite{BourgGolWenn1998, GolWenn2000}] In space 
dimension $D\ge 2$, there exists $0<C_D<C'_D$ such that
$$
\frac{C_D}{t}\le\Phi_r(t)\le\frac{C'_D}{t}\quad
	\hbox{ whenever }t>1\hbox{ and }\,0<r<\tfrac12\,.
$$
\end{Thm}

The lower bound and the upper bound in this theorem are obtained by very different means.

The upper bound follows from a Fourier series argument which is reminiscent of Siegel's prood of 
the classical Minkowski convex body theorem (see \cite{Siegel1936, Montgomery1994}.)

The lower bound, on the other hand, is obtained by working in physical space. Specifically, one uses 
a channel technique, introduced independently by P. Bleher \cite{Bleher1992} for the diffusive scaling.

This lower bound alone has an important consequence:

\begin{Cor}
For each $r>0$, the average of the free path length (mean free path) under the probability measure 
$\mu_r$ is infinite:
$$
\int_{(Z_r/\bZ^D)\times\bS^{D-1}}\tau_r(x,v)d\mu_r(x,v)=+\infty\,.
$$
\end{Cor}

\begin{proof}
Indeed, since $\Phi_r$ is the distribution of $\tau_r$ under $\mu_r$, one has
$$
\int_{(Z_r/\bZ^D)\times\bS^{D-1}}\tau_r(x,v)d\mu_r(x,v)
	=\int_0^\infty\Phi_r(t)dt\ge\int_1^\infty\frac{C_D}{t}dt=+\infty\,.
$$
\end{proof}

Recall that the average of the free path length unded the ``other" natural probability measure $\nu_r$ 
is precisely Santal\'o's formula for the mean free path:
$$
\ell_r=\iint_{\Gamma^+_r/\bZ^D}\tau_r(x,v)d\nu_r(x,v)=\frac{1-|\bB^D|r^D}{|\bB^{D-1}|r^{D-1}}\,.
$$
One might wonder why averaging the free path length $\tau_r$ under the measures $\nu_r$ and 
$\mu_r$ actually gives two so different results. 

First observe that Santal\'o's formula gives the mean free path under the probability measure $\nu_r$ 
concentrated on the surface of the obstacles, and is therefore irrelevant for particles that have not yet 
encountered an obstacle.

Besides, by using the lemma that implies Santal\'o's formula with $f(z)=\tfrac12z^2$, one has
$$
\iint_{(Z_r/\bZ^D)\times\bS^{D-1}}\tau_r(x,v)d\mu_r(x,v)
	=\frac1{\ell_r}\int_{\Gamma^+_r/\bZ^D}\tfrac12\tau_r(x,v)^2d\nu_r(x,v)\,.
$$

Whenever the components $v_1,\ldots,v_D$ are independent over $\bQ$, the linear flow in the 
direction $v$ is topologically transitive and ergodic on the $D$-torus, so that $\tau_r(x,v)<+\infty$ 
for each $r>0$ and $x\in\bR^D$. On the other hand, $\tau_r(x,v)=+\infty$ for some $x\in\ Z_r$ 
(the periodic billiard table) whenever $v$ belongs to some specific class of unit vectors whose 
components are rationally dependent, a class that becomes dense in $\bS^{D-1}$ as $r\to 0^+$. 
Thus, $\tau_r$ is strongly oscillating (finite for irrational directions, possibly infinite for a class of 
rational directions that becomes dense as $r\to 0^+$),  and this explains why $\tau_r$ doesn't 
have a second moment under $\nu_r$.

\smallskip
\noindent
\textit{Proof of the Bourgain-Golse-Wennberg lower bound}

We shall restrict our attention to the case of space dimension $D=2$.

As mentionned above, there are \textit{infinite open strips} included in $Z_r$ --- i.e. never meeting 
any obstacle. Call \textit{a channel} any such  open strip of maximum width, and let $\cC_r$ be the 
set of all channels included in $Z_r$.

If $S\in\cC_r$ and $x\in S$, define $\tau_S(x,v)$ the exit time from the channel starting from $x$ in 
the direction $v$, defined as
$$
\tau_S(x,v)=\inf\{t>0\,|\,x+tv\in\d S\}\,,\quad (x,v)\in S\times\bS^1\,.
$$
Obviously, any particle starting from $x$ in the channel $S$ in the direction $v$ must exit $S$ before 
it hits an obstacle (since no obstacle intersects $S$). Therefore
$$
\tau_r(x,v)\ge\sup\{\tau_S(x,v)\,|\,S\in\cC_r\hbox{ s.t. }x\in S\}\,,
$$
so that
$$
\Phi_r(t)\ge\mu_r\left(\bigcup_{S\in\cC_r}\{(x,v)\in (S/\bZ^2)\times\bS^1
	\,|\,\tau_S(x,v)>t/r\}\right)\,.
$$
This observation suggests that one should carefully study the set of channels $\cC_r$.

\smallskip
\noindent
\underline{Step 1: description of $\cC_r$.}  Given $\om\in\bS^1$, we define
$$
\cC_r(\om):=\{\hbox{channels of direction $\om$ in $\cC_r$}\}\,;
$$
We begin with a lemma which describes the structure of $\cC_r(\om)$.

\begin{Lem} Let $r\in[0,\tfrac12)$ and $\om\in\bS^1$. Then

\smallskip
\noindent
1) if $S\in\cC_r(\om)$, then 
$$
\cC_r(\om):=\{S+k\,|\,k\in\bZ^2\}\,;
$$
2) if $\cC_r(\om)\not=\varnothing$, then
$$
\om=\frac{(p,q)}{\sqrt{p^2+q^2}}
$$
with
$$
(p,q)\in\bZ^2\setminus\{(0,0)\}\hbox{ such that }
\hbox{g.c.d.}(p,q)=1\hbox{ and }\sqrt{p^2+q^2}<\frac1{2r}\,.
$$

We henceforth denote by $\cA_r$ the set of all such $\om\in\bS^1$. Then

\smallskip
\noindent
3) for $\om\in\cA_r$, the elements of $\cC_r(\om)$ are open strips of width
$$
w(\om,r)=\frac1{\sqrt{p^2+q^2}}-2r\,.
$$
\end{Lem}

\begin{proof}[Proof of the Lemma] Statement 1) is obvious.

As for statement 2), if $L$ is an infinite line of direction $\om\in\bS^1$ such that  $\om_2/\om_1$ is 
irrational, then $L/\bZ^2$ is an orbit of a linear flow on $\bT^2$ with irrational slope $\om_2/\om_1$. 
Therefore $L/\bZ^2$ is dense in $\bT^2$ so that $L$ cannot be included in $Z_r$.

Assume that 
$$
\om=\frac{(p,q)}{\sqrt{p^2+q^2}}\hbox{ with }
	(p,q)\in\bZ^2\setminus\{(0,0)\}\hbox{ coprime,}
$$
and let $L,L'$ be two infinite lines with direction $\om$, with equations
$$
qx-py=a\hbox{ and }qx-py=a'\hbox{ respectively.}
$$
Obviously
$$
\hbox{dist}(L,L')=\frac{|a-a'|}{\sqrt{p^2+q^2}}\,.
$$

\begin{figure}

\includegraphics[width=12.0cm]{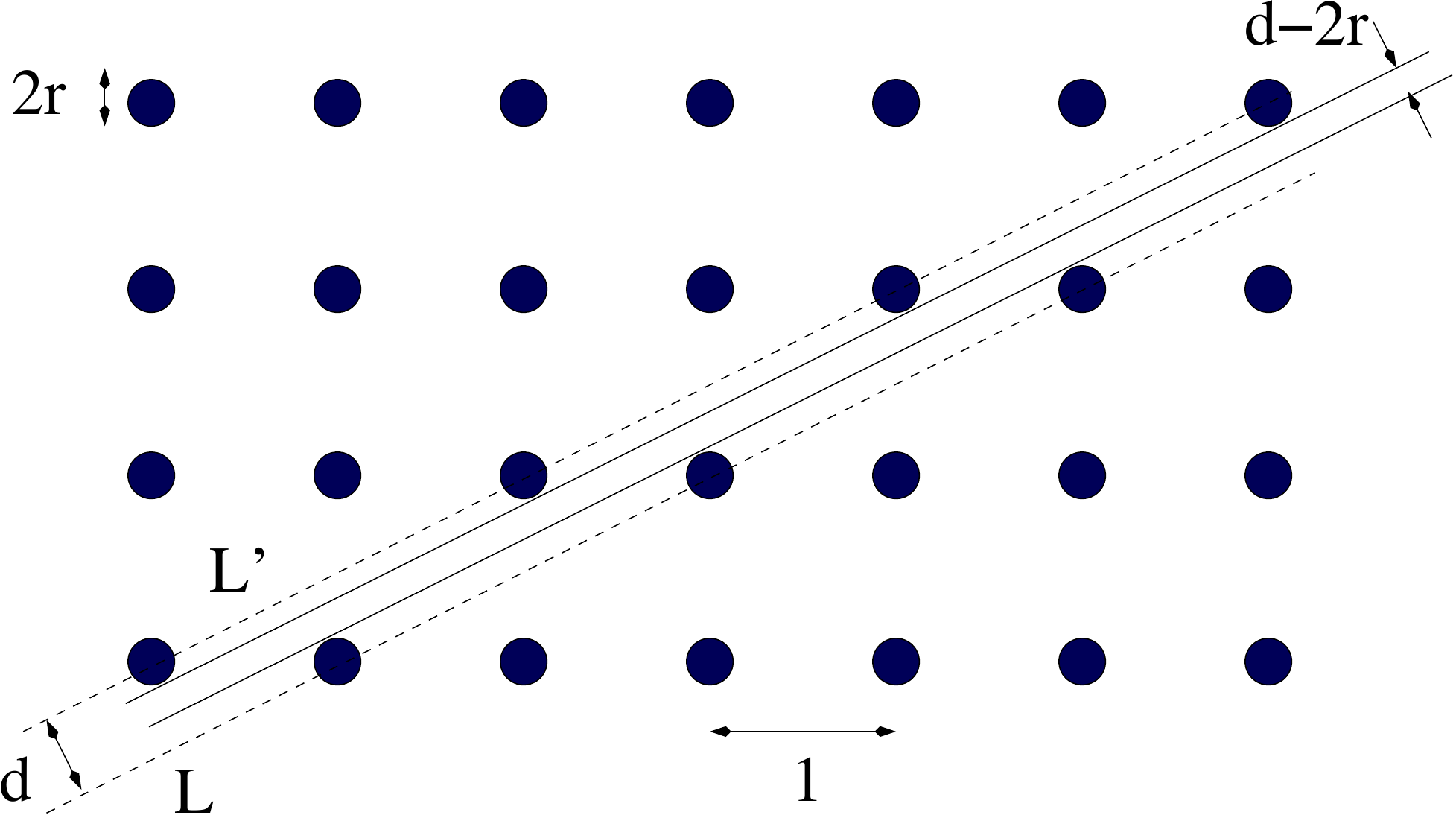}
\caption{A channel of direction $\om=\tfrac1{\sqrt{5}}(2,1)$; minimal distance $d$ between lines $L$ 
and $L'$ of direction $\om$ through lattice points}

\end{figure}

If $L\cup L'$ is the boundary of a channel of direction 
$$
\om=\tfrac{(p,q)}{\sqrt{p^2+q^2}}\in\cA_0
$$
included in $\bR^2\setminus\bZ^2$ --- i.e. of an element of $\cC_0(\om)$, then $L$ and $L'$ intersect 
$\bZ^2$ so that 
$$
a,a'\in p\bZ+q\bZ=\bZ
$$
--- the equality above following from the assumption that $p$ and $q$ are coprime.

Since $\hbox{dist}(L,L')>0$ is minimal, then $|a-a'|=1$, so that
$$
\hbox{dist}(L,L')=\frac1{\sqrt{p^2+q^2}}\,.
$$
Likewise, if $L\cup L'=\d S$ with $S\in\cC_r$, then $L$ and $L' $ are parallel infinite lines tangent to 
$\d Z_r$, and the minimal distance between any such distinct lines is
$$
\hbox{dist}(L,L')=\frac1{\sqrt{p^2+q^2}}-2r\,.
$$
This entails 2) and 3).
\end{proof}

\smallskip
\noindent
\underline{Step 2: the exit time from a channel.}  Let $\om=\tfrac{(p,q)}{\sqrt{p^2+q^2}}\in\cA_r$ and 
let $S\in\cC_r(\om)$. Cut $S$ into three parallel strips of equal width and call $\hat S$ the middle one. 
For each $t>1$ define
$$
\th\equiv\th(\om,r,t):=\arcsin\left(\frac{rw(\om,r)}{3t}\right)\,.
$$

\begin{Lem}
If $x\in\hat S$ and $v\in(R[-\th]\om,R[\th]\om)$, where $R[\th]$ designates the rotation of an angle $\th$, 
then
$$
\tau_S(x,v)\ge t/r\,.
$$
Moreover
$$
\mu_r((\hat S/\bZ^2)\times(R[-\th]\om,R[\th]\om))=\tfrac23 w(\om,r)\th(\om,r,t)\,.
$$
\end{Lem}

The proof of this lemma is perhaps best explained by considering Figure 11.

\begin{figure}

\includegraphics[width=12.0cm]{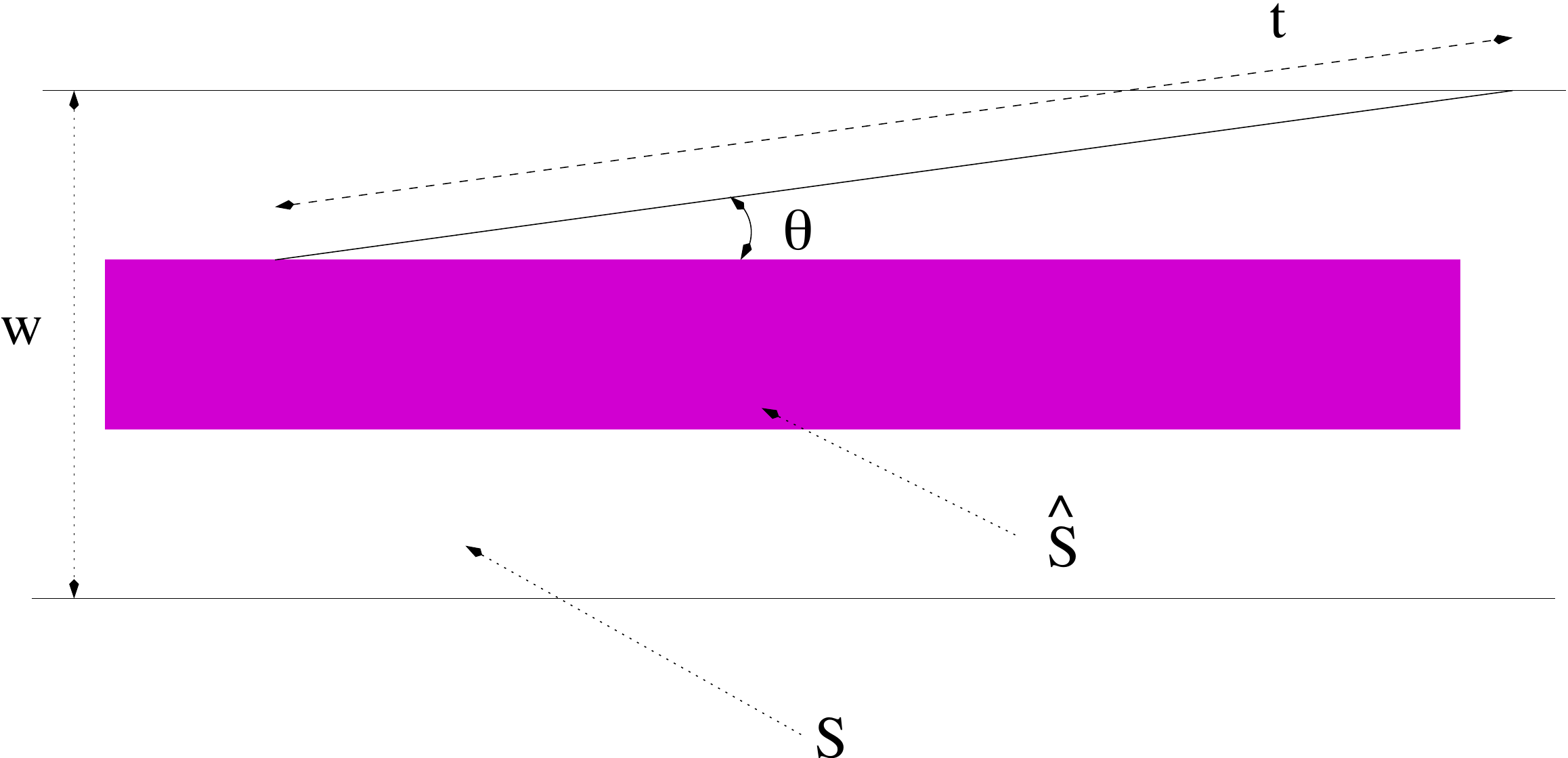}

\caption{Exit time from the middle third $\hat S$ of an infinite strip $S$ of width $w$}

\end{figure}

\smallskip
\noindent
\underline{Step 3: putting all channels together.} Recall that we need to estimate
$$
\mu_r\left(\bigcup_{S\in\cC_r}\{(x,v)\in (S/\bZ^2)\times\bS^1\,|\,\tau_S(x,v)>t/r\}\right)\,.
$$

Pick 
$$
\cA_r\ni\om=\tfrac{(p,q)}{\sqrt{p^2+q^2}}
	\not=\tfrac{(p',q')}{\sqrt{p'^2+q'^2}}=\om'\in\cA_r\,.
$$
Observe that
$$
\begin{aligned}
{}&|\sin(\widehat{\om,\om'})|=\tfrac{|pq'-p'q|}{\sqrt{p^2+q^2}\sqrt{p'^2+q'^2}}
	\ge \tfrac{1}{\sqrt{p^2+q^2}\sqrt{p'^2+q'^2}}
\\
&\ge\max\left(\tfrac{2r}{\sqrt{p^2+q^2}},\tfrac{2r}{\sqrt{p'^2+q'^2}}\right)
	\ge\sin\th(\om,r,t)+\sin\th(\om',r,t)
\\
&\ge\sin(\th(\om,r,t)+\th(\om',r,t))
\end{aligned}
$$
whenever $t>1$.

Then, whenever $S\in\cC_r(\om)$ and $S'\in\cC_r(\om')$
$$
(\hat S\times(R[-\th]\om,R[\th]\om)))\cap(\hat S'\times(R[\th']\om',R[\th']\om')))
=\varnothing
$$
with $\th=\th(\om,r,t)$, $\th'=\th'(\om',r,t)$ and $R[\th]=$the rotation of an angle $\th$.

Moreover, if $\om=\tfrac{(p,q)}{\sqrt{p^2+q^2}}\in\cA_r$ then 
$$
|\hat S/\bZ^2|=\tfrac13 w(\om,r)\sqrt{p^2+q^2}\,,
$$
while
$$
\#\{S/\bZ^2\,|\,S\in\cC_r(\om)\}=1\,.
$$

\begin{figure}

\includegraphics[width=8.0cm]{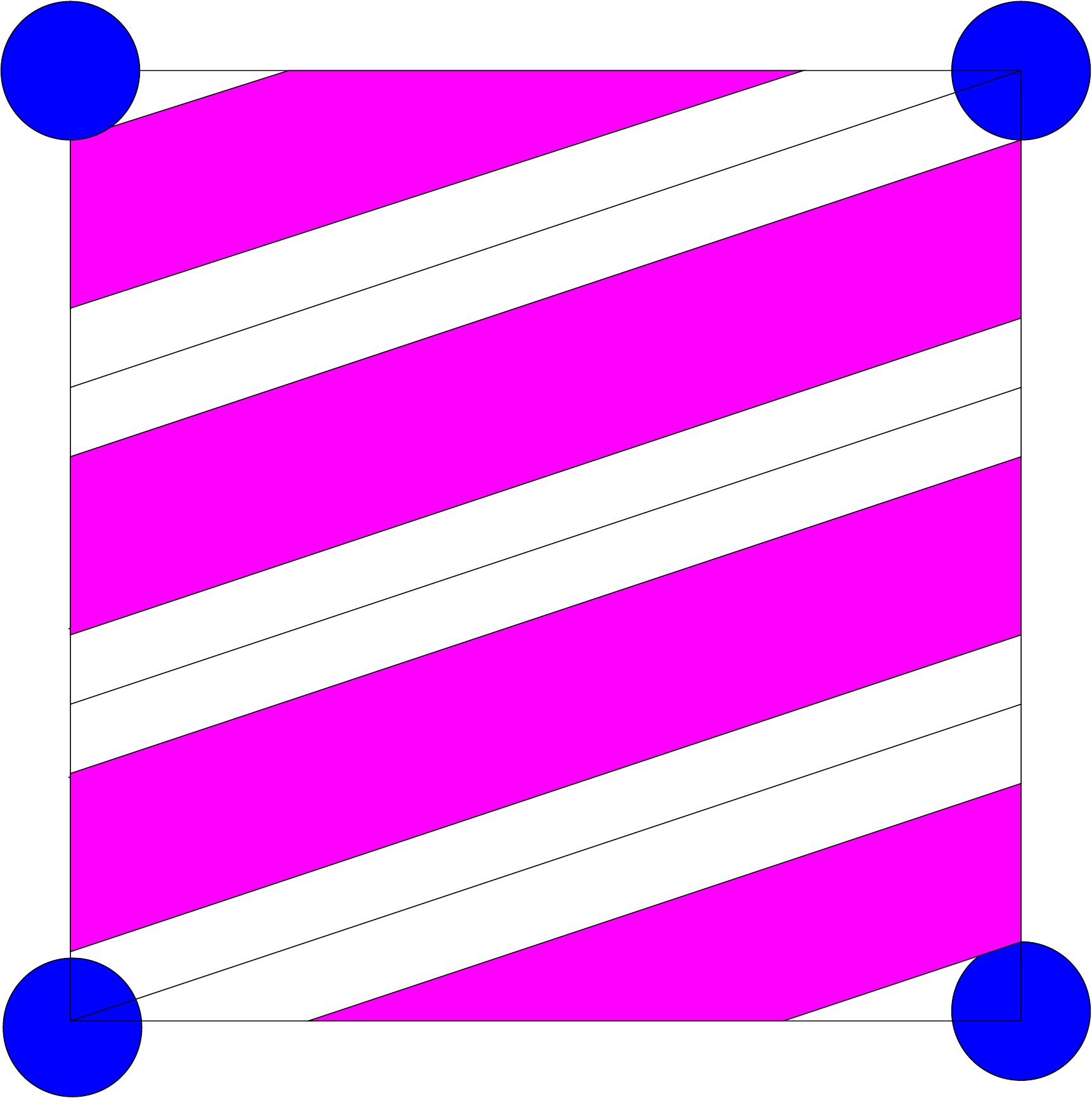}

\centerline{A channel modulo $\bZ^2$}

\end{figure}

\smallskip
\noindent
\underline{Conclusion:} Therefore, whenever $t>1$
$$
\begin{aligned}
{}&\bigcup_{S\in\cC_r}(\hat S/\bZ^2)\times(R[-\th]\om,R[\th]\om)
\\
&\qquad\subset\bigcup_{S\in\cC_r}
	\{(x,v)\in (S/\bZ^2)\times\bS^1\,|\,\tau_S(x,v)>t/r\}\,.
\end{aligned}
$$
and the left-hand side is a disjoint union. Hence
$$
\begin{aligned}
\mu_r\left(\bigcup_{S\in\cC_r}\{(x,v)\in (S/\bZ^2)\times\bS^1
	\,|\,\tau_S(x,v)>t/r\}\right)
\\
\ge
\sum_{\om\in\cA_r}\mu_r((\hat S/\bZ^2)\times(R[-\th]\om,R[\th]\om))
\\
=
\sum_{g.c.d.(p,q)=1\atop p^2+q^2<1/4r^2}
	\tfrac13w(\om,r)\sqrt{p^2+q^2}\cdot 2\th(\om,r,t)
\\
=
\sum_{g.c.d.(p,q)=1\atop p^2+q^2<1/4r^2}
\tfrac23\sqrt{p^2+q^2}w(\om,r)\arcsin\left(\frac{rw(\om,r)}{3t}\right)
\\
\ge
\sum_{g.c.d.(p,q)=1\atop p^2+q^2<1/4r^2}
\tfrac23\sqrt{p^2+q^2}\frac{rw(\om,r)^2}{3t}\,.
\end{aligned}
$$

Now
$$
\sqrt{p^2+q^2}<1/4r\Rightarrow w(\om,r)
	=\tfrac1{\sqrt{p^2+q^2}}-2r\ge\tfrac1{2\sqrt{p^2+q^2}}\,,
$$
so that, eventually
$$
\begin{aligned}
\Phi_r(t)\ge\sum_{g.c.d.(p,q)=1\atop p^2+q^2<1/16r^2}
	\tfrac23\sqrt{p^2+q^2}\frac{rw(\om,r)^2}{3t}
\\
\ge\frac{r^2}{18t}\sum_{g.c.d.(p,q)=1\atop p^2+q^2<1/16r^2}
	\left[\frac{1}{r\sqrt{p^2+q^2}}\right]\,.	
\end{aligned}
$$
This gives the desired conclusion since
$$
\sum_{g.c.d.(p,q)=1\atop p^2+q^2<1/16r^2}\left[\frac{1}{4r\sqrt{p^2+q^2}}\right]
=
\sum_{p^2+q^2<1/16r^2}1\sim\frac{\pi}{16r^2}\,.
$$
The first equality is proved as follows: the term
$$
\left[\frac{1}{4r\sqrt{p^2+q^2}}\right]
$$
is the number of integer points on the segment of length $1/4r$ in the direction $(p,q)$ with $(p,q)\in\bZ^2$  
such that $g.c.d.(p,q)=1$.
\begin{figure}

\includegraphics[width=10.0cm]{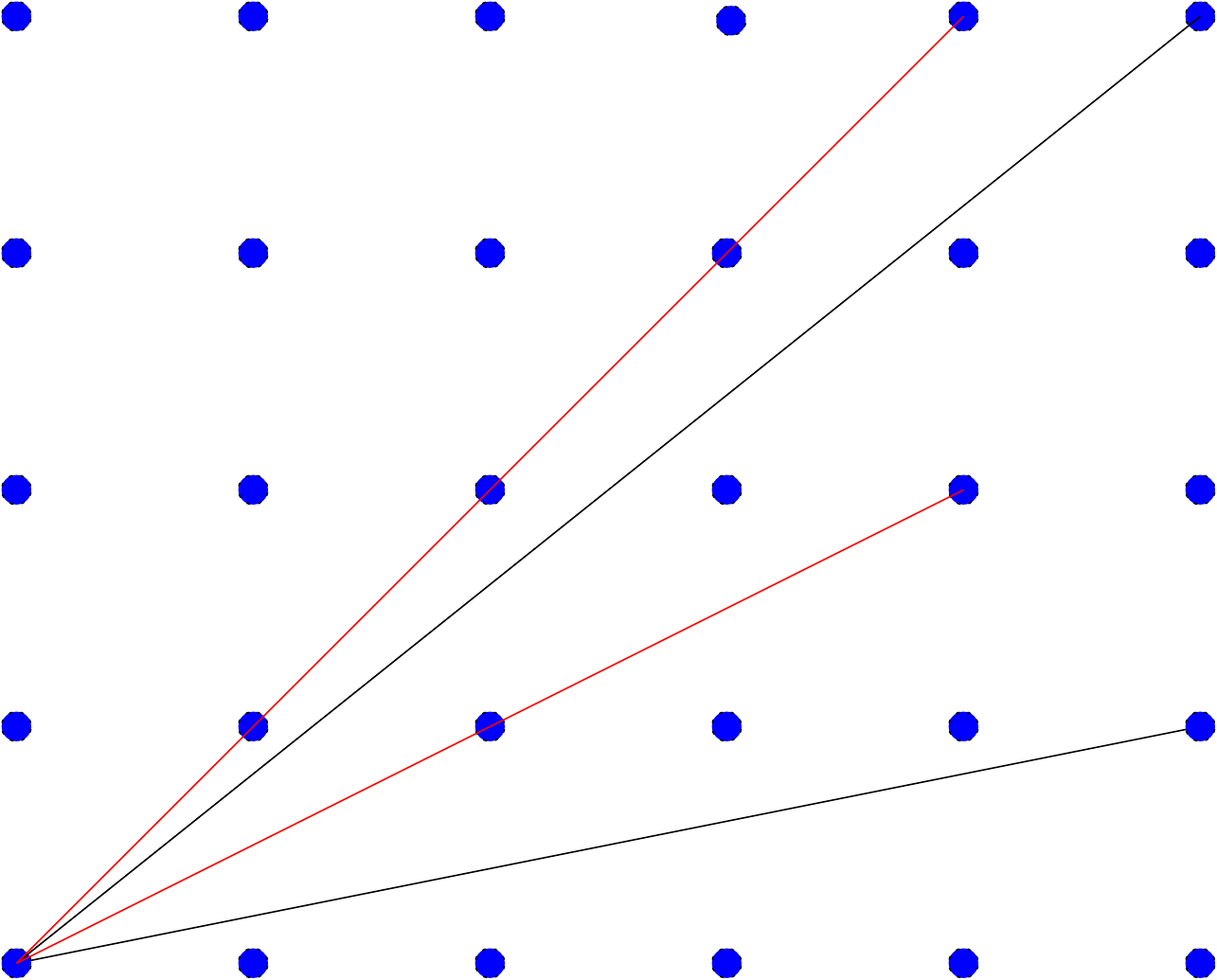}

\caption{Black lines issued from the origin terminate at integer points with coprime coordinates; red lines 
terminate at integer points whose coordinates are not coprime}

\end{figure}

\bigskip
The Bourgain-Golse-Wennberg theorem raises the question, of whe- ther $\Phi_r(t)\simeq C/t$ in some 
sense as $r\to 0^+$ and $t\to+\infty$. Given the very different nature of the arguments used to establish 
the upper and the lower bounds in that theorem, this is a highly nontrivial problem, whose answer seems 
to be known only in space dimension $D=2$ so far. We shall return to this question later, and see that the 
$2$-dimensional situation is amenable to a class of very specific techniques based on continued fractions, 
that can be used to encode particle trajectories of the periodic Lorentz gas.

\smallskip
A first answer to this question, in space dimension $D=2$, is given by the following

\begin{Thm}[Caglioti-Golse \cite{CagliotiFG2003}]
Assume $D=2$ and define, for each $v\in\bS^1$,
$$
\phi_r(t|v)=\mu_r(\{x\in Z_r/\bZ^2\,|\,\tau_r(x,v)\ge t/r\}\,,\quad t\ge 0\,.
$$
Then there exists $\Phi:\,\bR_+\to\bR_+$ such that
$$
\frac1{|\ln\eps|}\int_\eps^{1/4}\phi_r(t,v)\frac{dr}{r}\to\Phi(t)\hbox{ a.e. in }v\in\bS^1
$$
in the limit as $\eps\to 0^+$. Moreover
$$
\Phi(t)\sim\frac1{\pi^2t}\hbox{ as }t\to+\infty\,.
$$
\end{Thm}

Shortly after \cite{CagliotiFG2003}  appeared, F. Boca and A. Zaharescu improved our method
and managed to compute $\Phi(t)$ explicitly for each $t\ge 0$. One should keep in mind that
their formula had been conjectured earlier by P. Dahlqvist \cite{Dahl}, on the basis of a formal
computation. 

\begin{Thm}[Boca-Zaharescu \cite{BocaZaha2007}]\label{T-BocaZaha}For each $t>0$
$$
\Phi_r(t)\to\Phi(t)=\int_t^\infty(s-t)g(s)ds
$$
in the limit as $r\to 0^+$, where
$$
g(s)=\tfrac{24}{\pi^2}\times\left\{
\begin{matrix} 1\quad &s\in[0,1]\,,\\
\tfrac1s+2\left(1-\tfrac1s\right)^2\ln(1-\tfrac1s)
	-\tfrac12\left|1-\tfrac2s\right|^2\ln|1-\tfrac2s|\quad &s\in(1,\infty)\,.
\end{matrix}\right.
$$
\end{Thm}

\begin{figure}

\includegraphics[width=6.8cm]{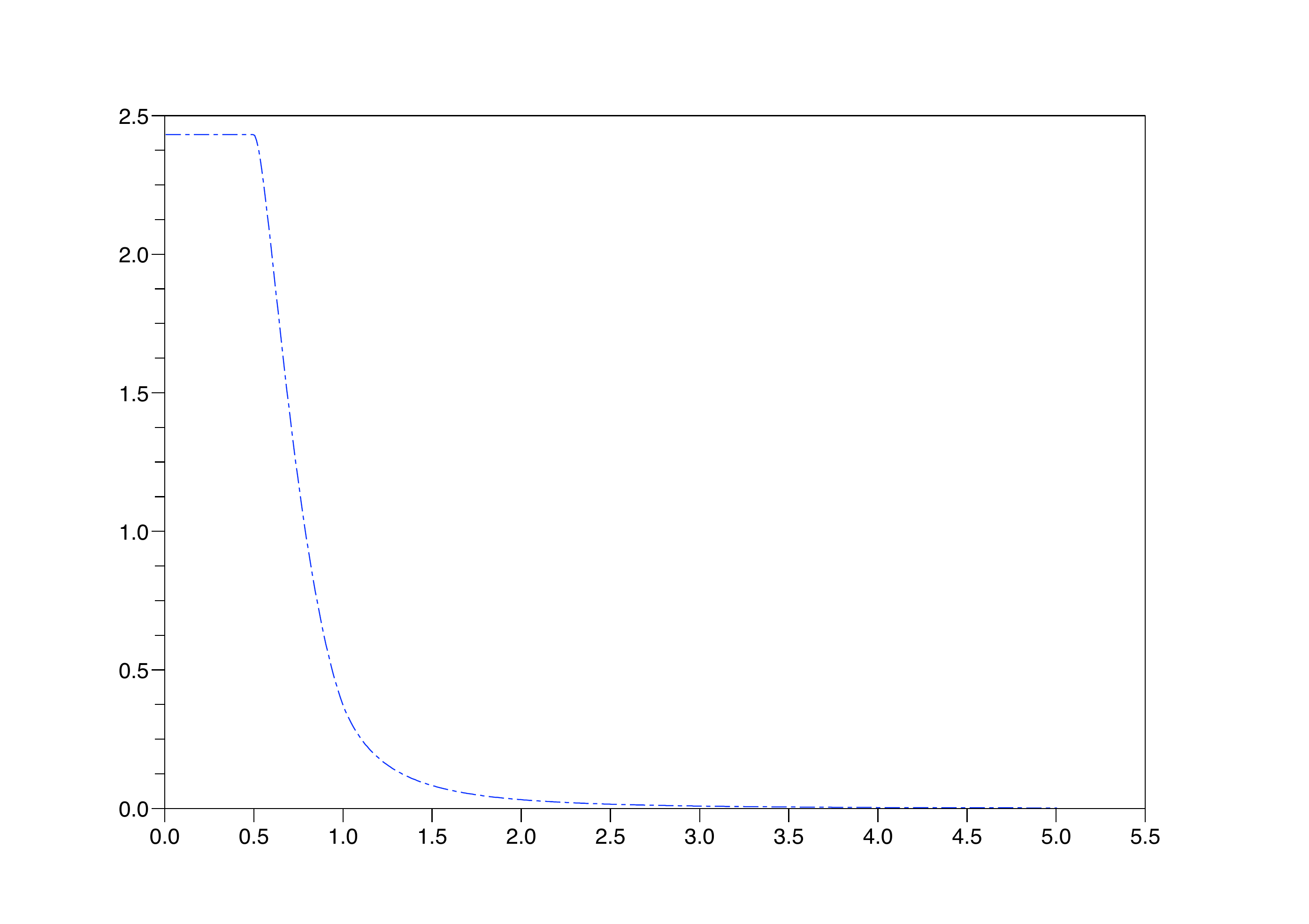}\includegraphics[width=6.8cm]{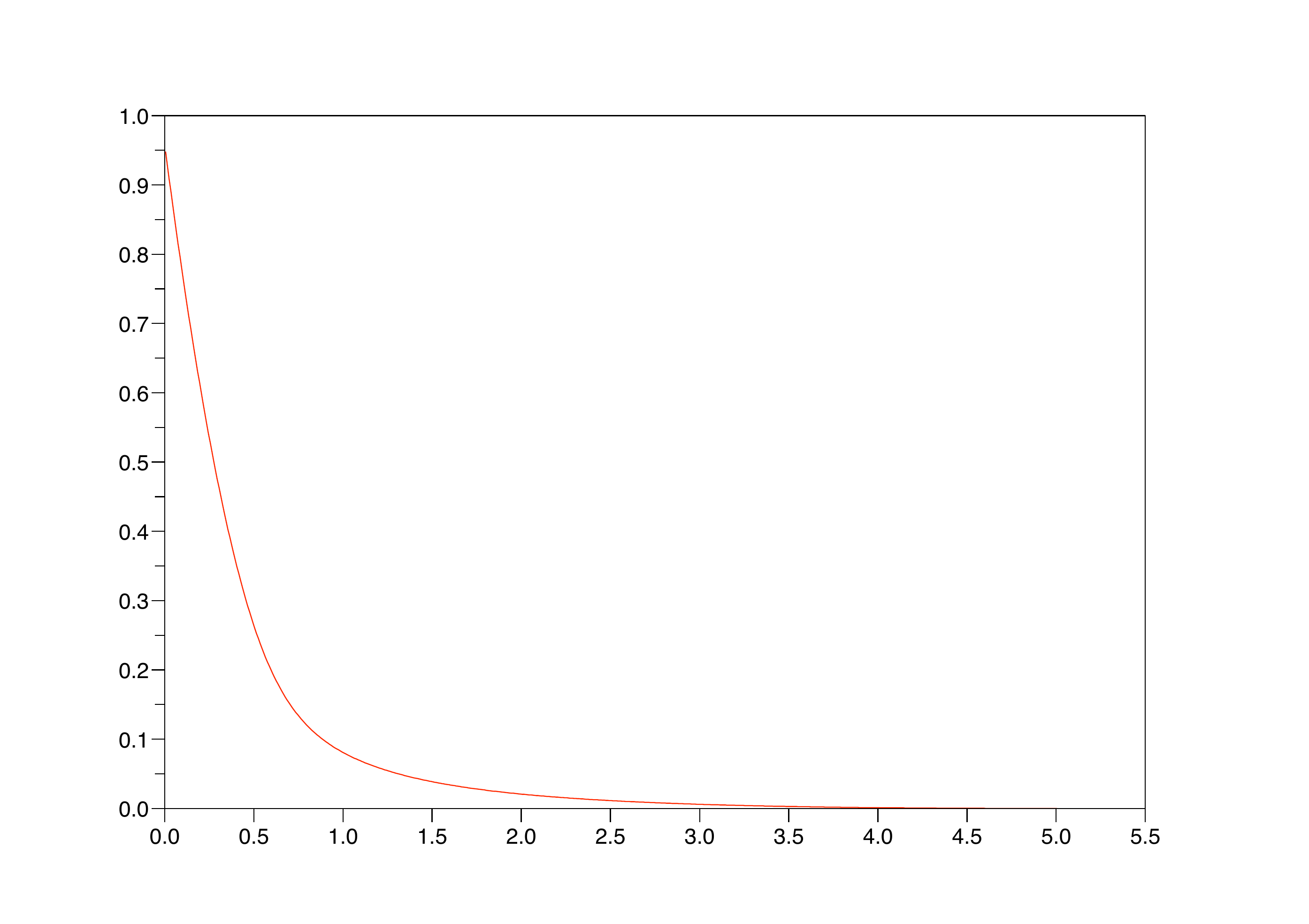}

\caption{Graph of $\Phi(t)$ (right) and $g(t)=\Phi''(t)$ (left)}

\end{figure}

In the sequel, we shall return to the continued and Farey fractions techniques used in the proofs
of these two results, and generalize them.


\section{A negative result for the Boltzmann-Grad limit\\ of the periodic Lorentz gas}


The material at our disposal so far provides us with a first answer --- albeit a negative one --- to the 
problem of determining the Boltzmann-Grad limit of the periodic Lorentz gas. 

For simplicity, we consider the case of a Lorentz gas enclosed in a periodic box 
$\bT^D=\bR^D/\bZ^D$ of unit side. The distance between neighboring obstacles is supposed 
to be $\eps^{D-1}$ with $0<\eps=1/n$, for $n\in\bN$ and $n>2$ so that $\eps<1/2$, while the 
obstacle radius is $\eps^D<\tfrac12\eps^{D-1}$ --- so that obstacles never overlap. Define
$$
Y_\eps=\{x\in\bT^D\,|\,\hbox{dist}(x,\eps^{D-1}\bZ^D)>\eps^D\}
	=\eps^{D-1}(Z_\eps/\bZ^D)\,.
$$

For each $f^{in}\in C(\bT^D\times\bS^{D-1})$, let $f_\eps$ be the solution of 
$$
\begin{aligned}
\d_tf_\eps+v\cdot\grad_xf_\eps=0\,,\quad &(x,v)\in Y_\eps\times\bS^{D-1}
\\
f_\eps(t,x,v)=f_\eps(t,x,\cR[n_x]v)\,,\quad &(x,v)\in\d Y_\eps\times\bS^{D-1}
\\
f_\eps\rstr_{t=0}=f^{in}\,,\quad&
\end{aligned}
$$
where $n_x$ is unit normal vector to $\d Y_\eps$ at the point $x$, pointing towards the interior of 
$Y_\eps$. 

By the method of characteristics
$$
f_\eps(t,x,v)=
f^{in}\left(\eps^{D-1}X_\eps\left(-\tfrac{t}{\eps^{D-1}};\tfrac{x}{\eps^{D-1}},v\right);
	V_\eps\left(-\tfrac{t}{\eps^{D-1}};\tfrac{x}{\eps^{D-1}},v\right)\right)
$$
where $(X_\eps,V_\eps)$ is the billiard flow in $Z_\eps$.

The main result in this section is the following

\begin{Thm}[Golse \cite{Golse2003, Golse2008}]
There exist initial data $f^{in}\equiv f^{in}(x)\in C(\bT^D)$ such that no subsequence of $f_\eps$ 
converges for the weak-* topology of $L^\infty(\bR_+\times\bT^D\times\bS^{D-1})$ to the solution 
$f$ of a linear Boltzmann equation of the form
$$
\begin{aligned}
(\d_t+v\cdot\grad_x)f(t,x,v)&=\si\int_{\bS^{D-1}}p(v,v')(f(t,x,v')-f(t,x,v))dv'
\\
f\rstr_{t=0}&=f^{in}\,,
\end{aligned}
$$
where $\si>0$ and $0\le p\in L^2(\bS^{D-1}\times\bS^{D-1})$ satisfies
$$
\int_{\bS^{D-1}}p(v,v')dv'=\int_{\bS^{D-1}}p(v',v)dv'=1\hbox{ a.e. in }v\in\bS^{D-1}\,.
$$
\end{Thm}

This theorem has the following important --- and perhaps surprising --- consequence: \textit{the 
Lorentz kinetic equation cannot govern the Boltzmann-Grad limit of the particle density in the case 
of a periodic distribution of obstacles}.

\begin{proof} 
The proof of the negative result above involves two different arguments:

a) the existence of a spectral gap for any linear Boltzmann equation, and

b) the lower bound for the distribution of free path lengths in the Bourgain-Golse-Wennberg theorem.

\smallskip
\noindent
\underline{Step 1: Spectral gap for the linear Boltzmann equation:}

With $\si>0$ and $p$ as above, consider the unbounded operator $A$ on $L^2(\bT^D\times\bS^{D-1})$
defined by
$$
(A\phi)(x,v)=
	-v\cdot\grad_x\phi(x,v)-\si\phi(x,v)+\si\int_{\bS^{D-1}}p(v,v')\phi(x,v')dv'\,,
$$
with domain
$$
D(A)=\{\phi\in L^2(\bT^D\times\bS^{D-1})\,|
	\,v\cdot\grad_x\phi\in L^2(\bT^D\times\bS^{D-1})\}\,.
$$
Then

\begin{Thm}[Ukai-Point-Ghidouche \cite{UkaiPointGhid1979}] 
There exists positive constants $C$ and $\g$ such that
$$
\|e^{tA}\phi-\la\phi\ra\|_{L^2(\bT^D\times\bS^{D-1})}
    \le Ce^{-\g t}\|\phi\|_{L^2(\bT^D\times\bS^{D-1})}\,,\quad t\ge 0\,,
$$
for each $\phi\in L^2(\bT^D\times\bS^{D-1})$, where
$$
\la\phi\ra=\tfrac1{|\bS^{D-1}|}\iint_{\bT^D\times\bS^{D-1}}\phi(x,v)dxdv\,.
$$
\end{Thm}

Taking this theorem for granted, we proceed to the next step in the proof, leading to an explicit lower 
bound for the particle density.

\smallskip
\noindent
\underline{Step 2: Comparison with the case of absorbing obstacles}

Assume that $f^{in}\equiv f^{in}(x)\ge 0$ on $\bT^D$. Then
$$
f_\eps(t,x,v)\ge g_\eps(t,x,v)=f^{in}(x-tv)\indc_{Y_\eps}(x)
	\indc_{\eps^{D-1}\tau_\eps(x/\eps^{D-1},v)>t}\,.
$$
Indeed, $g$ is the density of particles with the \textit{same} initial data as $f$, but assuming that each 
particle \textit{disappear} when colliding with an obstacle instead of being reflected.

Then
$$
\indc_{Y_\eps}(x)\to 1\hbox{ a.e. on $\bT^D$ and }|\indc_{Y_\eps}(x)|\le 1
$$
while, after extracting a subsequence if needed,
$$
\indc_{\eps^{D-1}\tau_\eps(x/\eps^{D-1},v)>t}
	\rightharpoonup\Psi(t,v)\hbox{ in }L^\infty(\bR_+\times\bT^D\times\bS^{D-1})
	\hbox{ weak-*}\,.
$$
Therefore,  if $f$ is a weak-* limit point of $f_\eps$ in $L^\infty(\bR_+\times\bT^D\times\bS^{D-1})$ as 
$\eps\to 0$
$$
f(t,x,v)\ge f^{in}(x-tv)\Psi(t,v)\hbox{ for a.e. }(t,x,v)\,.
$$

\smallskip
\noindent
\underline{Step 3: using the lower bound on the distribution of $\tau_r$}

Denoting by $dv$ the uniform probability measure on $\bS^{D-1}$
$$
\begin{aligned}
\tfrac1{|\bS^{D-1}|}\iint_{\bT^D\times\bS^{D-1}}&f(t,x,v)^2dxdv
\\
&\ge
\tfrac1{|\bS^{D-1}|}\iint_{\bT^D\times\bS^{D-1}}f^{in}(x-tv)^2\Psi(t,v)^2dxdv
\\
&=
\int_{\bT^D}f^{in}(y)^2dy\tfrac1{|\bS^{D-1}|}\int_{\bS^{D-1}}\Psi(t,v)^2dv
\\
&\ge
\|f^{in}\|^2_{L^2(\bT^D)}\left(\tfrac1{|\bS^{D-1}|}\int_{\bS^{D-1}}\Psi(t,v)dv\right)^2
\\
&=
\|f^{in}\|^2_{L^2(\bT^D)}\Phi(t)^2\,.
\end{aligned}
$$
By the Bourgain-Golse-Wennberg lower bound on the distribution $\Phi$ of free path lengths
$$
\|f(t,\cdot,\cdot)\|_{L^2(\bT^D\times\bS^{D-1})}\ge 
	\frac{C_D}{t}\|f^{in}\|_{L^2(\bT^D)}\,,\quad t>1\,.
$$
On the other hand, by the spectral gap estimate, if $f$ is a solution of the linear Boltzmann equation, 
one has
$$
\|f(t,\cdot,\cdot)\|_{L^2(\bT^D\times\bS^{D-1})}
	\le\int_{\bT^D}f^{in}(y)dy+Ce^{-\g t}\|f^{in}\|_{L^2(\bT^D)}
$$
so that
$$
\frac{C_D}{t}\le\frac{\|f^{in}\|_{L^1(\bT^D)}}{\|f^{in}\|_{L^2(\bT^D)}}+Ce^{-\g t}
$$
for each $t>1$.

\smallskip
\noindent
\underline{Step 4: choice of initial data}

Pick $\rho$ to be a bump function supported near $x=0$ and such that
$$
\int\rho(x)^2dx=1\,.
$$
Take $f^{in}$ to be $x\mapsto\l^{D/2}\rho(\l x)$ periodicized, so that
$$
\int_{\bT^D}f^{in}(x)^2dx=1\,,
	\hbox{ while }\int_{\bT^D}f^{in}(y)dy=\l^{-D/2}\int\rho(x)dx\,.
$$
For such initial data, the inequality above becomes
$$
\frac{C_D}{t}\le\l^{-D/2}\int\rho(x)dx+Ce^{-\g t}\,.
$$
Conclude by choosing $\l$ so that
$$
\l^{-D/2}\int\rho(x)dx<\sup_{t>1}\left(\frac{C_D}{t}-Ce^{-\g t}\right)>0\,.
$$
\end{proof}

\smallskip
\noindent
\textbf{Remarks:}

\noindent
1) The same result (with the same proof) holds for any smooth obstacle shape included in a shell 
$$
\{x\in\bR^D\,|\,C\eps^D<\hbox{dist}(x,\eps^{D-1}\bZ^D)<C'\eps^D\}\,.
$$

\noindent
2) The same result (with same proof) holds if the specular reflection boundary condition is replaced 
by more general boundary conditions, such as absorption (partial or complete) of the particles at 
the boundary of the obstacles, diffuse reflection, or any convex combination of specular and diffuse 
reflection --- in the classical kinetic theory of gases, such boundary conditons are known as 
``accomodation boundary conditions".

\noindent
3) But introducing even the smallest amount of stochasticity in any periodic configuration of 
obstacles can again lead to a Boltzmann-Grad limit that is described by the Lorentz kinetic model.

\noindent
\textbf{Example.} (Wennberg-Ricci \cite{WennRicci2004}) In space dimension $2$,  take obstacles 
that are disks of radius $r$ centered at the vertices of the lattice $r^{1/(2-\eta)}\bZ^2$, assuming that 
$0<\eta<1$. Santal\'o's formula suggests that the free-path lengths scale like $r^{\eta/(2-\eta)}\to 0$.

Suppose the obstacles are removed independently with large probability --- specifically, with probability 
$p=1-r^{\eta/(2-\eta)}$. In that case, the Lorentz kinetic equation governs the 1-particle density in the 
Boltzmann-Grad limit as $r\to 0^+$.

\smallskip
Having explained why neither the Lorentz kinetic equation nor any linear Boltzmann equation can 
govern the Boltzmann-Grad limit of the periodic Lorentz gas, in the remaining part of these notes, we 
build the necessary material used in the description of that limit.


\section{Coding particle trajectories with continued fractions}


With the Bourgain-Golse-Wennberg lower bound for the distribution of free path lengths in the 
periodic Lorentz gas, we have seen that the $1$-particle phase space density is bounded below 
by a quantity that is incompatible with the spectral gap of any linear Boltzmann equation --- in 
particular with the Lorentz kinetic equation.

In order to further analyze the Boltzmann-Grad limit of the periodic Lorentz gas, we cannot content 
ourselves with even more refined estimates on the distribution of free path lengths, but we need a 
convenient way to encode particle trajectories.

More precisely, the two following problems must be answered somehow: 

\smallskip
\noindent
\underline{First problem:} for a particle leaving the surface of an obstacle in a given direction, to 
find the position of its next collision with an obstacle;

\smallskip
\noindent
\underline{Second problem:} average --- in some sense to be defined --- in order to eliminate the 
direction dependence.

\smallskip
From now on, our discussion is limited to the case of spatial dimension $D=2$, as we shall use 
continued fractions, a tool particularly well adapted to understanding the rational approximation 
of real numbers. Treating the case of a space dimension $D>2$ along the same lines would require
a better understanding of \textit{simultaneous} rational approximation  of $D-1$ real numbers (by 
$D-1$ rational numbers with the same denominator), a notoriously more difficult problem.

\smallskip
We first introduce some basic geometrical objects used in coding particle trajectories.

The first such object is the notion of \textit{impact parameter}.

For a particle with velocity $v\in\bS^1$ located at the position $x$ on the surface of an obstacle 
(disk of radius $r$), we define its impact parameter $h_r(x,v)$  by the formula
$$
h_r(x,v)=\sin(\widehat{n_x,v})\,.
$$
In other words, the absolute value of the impact parameter $h_r(x,v)$ is the distance of the center 
of the obstacle to the infinite line of direction $v$ passing through $x$ .

Obviously
$$
h_r(x,\cR[n_x]v)=h_r(x,v)
$$
where we recall the notation $\cR[n]v=v-2v\cdot nn$.

\begin{figure}

\includegraphics[width=10.0cm]{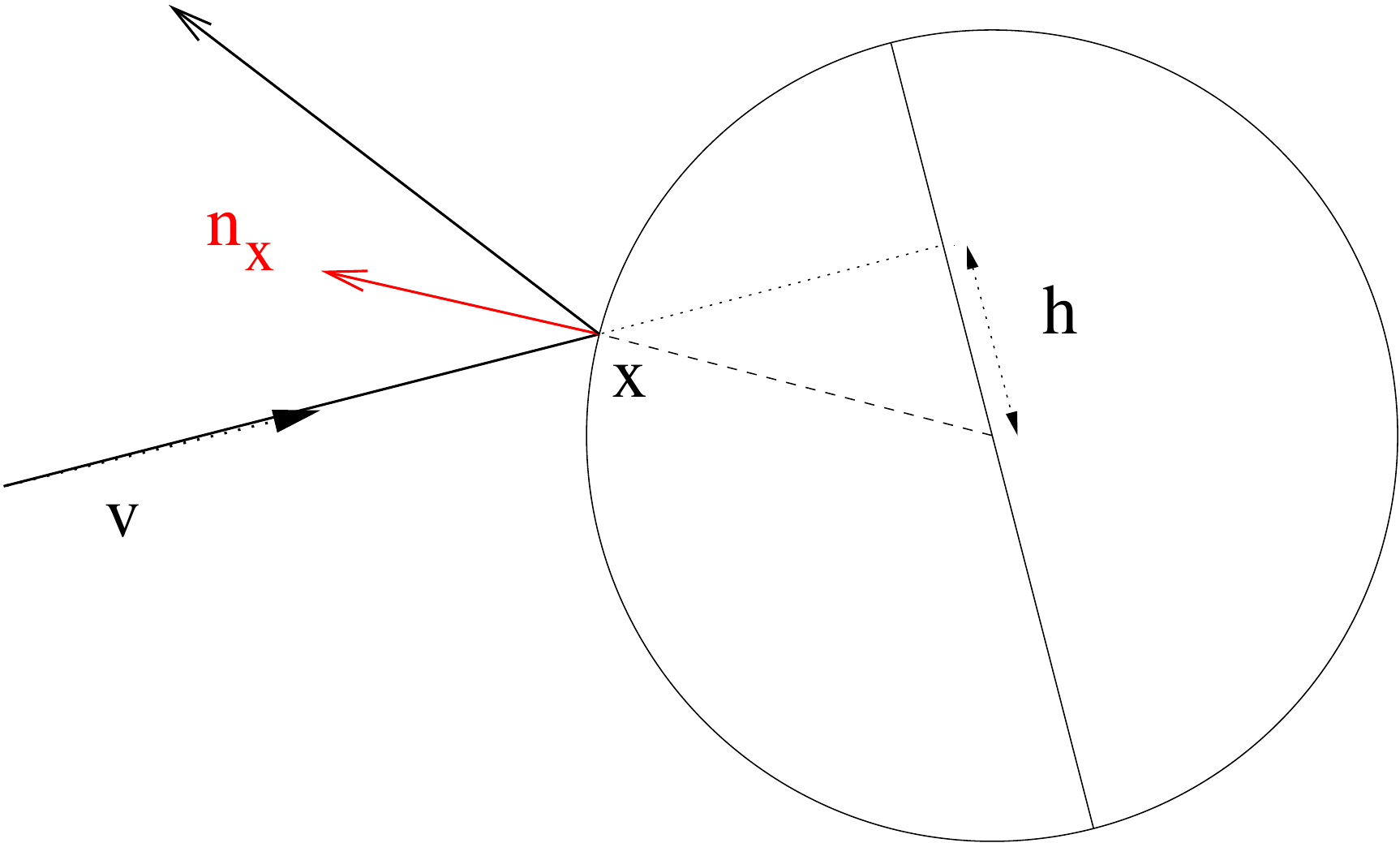}

\caption{The impact parameter $h$ corresponding with the collision point $x$ at the surface of an 
obstacle, and a direction $v$}

\end{figure}

The next important object in computing particle trajectories in the Lorentz gas is the \textit{transfer 
map}.

For a particle leaving the surface of an obstacle in the direction $v$ and with impact parameter $h'$, 
define
$$
T_r(h',v)=(s,h)\hbox{ with }
\left\{\begin{array}{l}
s=\hbox{ $r\times$ distance to the next collision point}
\\
h=\hbox{ impact parameter at the next collision}
\end{array}\right.
$$

Particle trajectories in the Lorentz gas are completely determined by the transfer map $T_r$ and its 
iterates.

Therefore, a first step in finding the Boltzmann-Grad limit of the periodic, $2$-dimensional Lorentz 
gas, is to compute the limit of $T_r$ as $r\to 0^+$, in some sense that will be explained later.

\begin{figure}

\includegraphics[width=10.0cm]{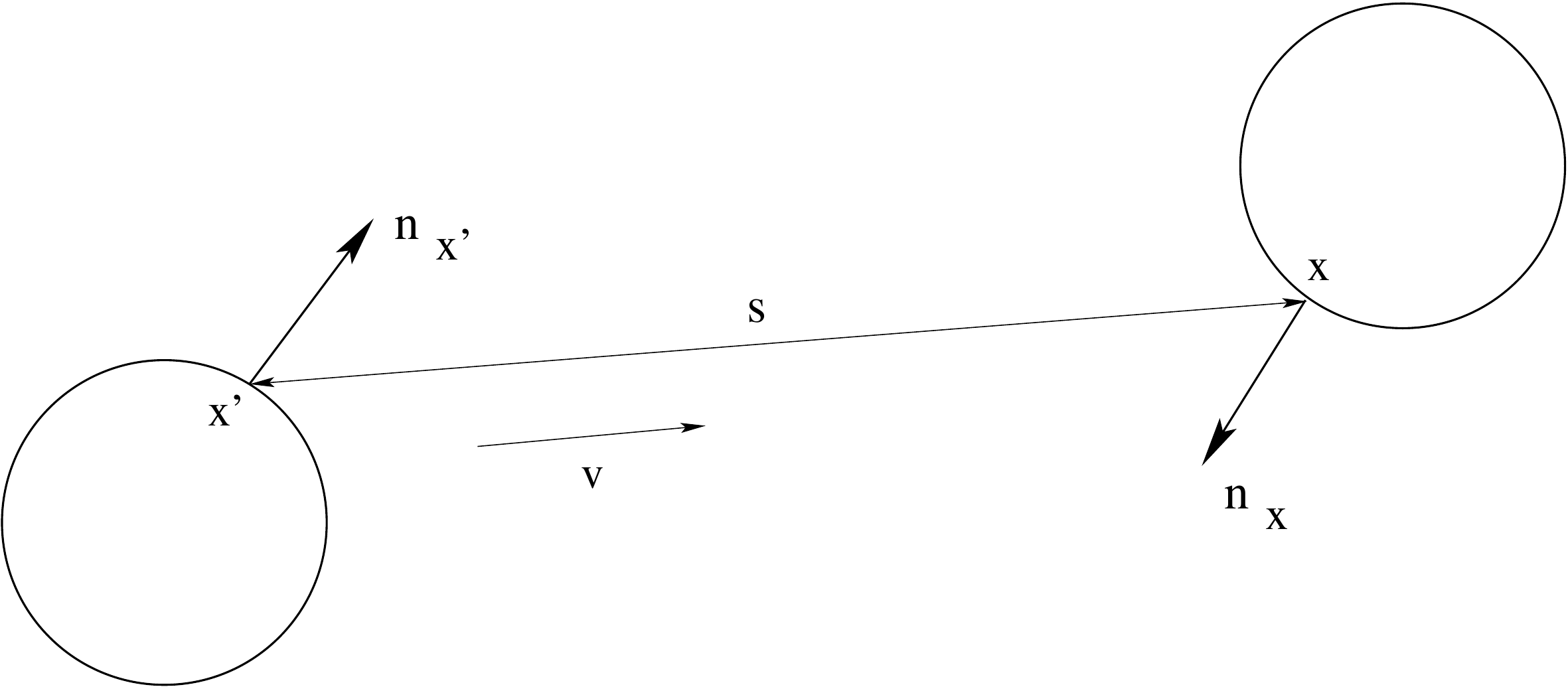}

\caption{The transfer map}

\end{figure}

At first sight, this seems to be a desperately hard problem to solve, as particle trajectories in the 
periodic Lorentz gas depend on their directions and the obstacle radius in the strongest possible 
way. Fortunately, there is an interesting property of rational approximation on the real line that 
greatly reduces the complexity of this problem.

\smallskip
\noindent
\textbf{The 3-length theorem}

\underline{\sc Question (R. Thom, 1989):} on a flat $2$-torus with a disk removed, consider a linear 
flow with irrational slope. What is the longest orbit?

\begin{Thm}[Blank-Krikorian \cite{BlankKriko1993}]
On a flat $2$-torus  with a segment removed, consider a linear flow with irrational slope $0<\a<1$.
The orbits of this flow have at most $3$ different lengths --- exceptionally $2$, but generically $3$. 
Moreover, in the generic case where these orbits have exactly $3$ different lengths, the length of
the longest orbit is the sum of the two other lengths.
\end{Thm}

These lengths are expressed in terms of the continued fraction expansion of the slope $\a$. 

\begin{figure}

\includegraphics[width=9.0cm]{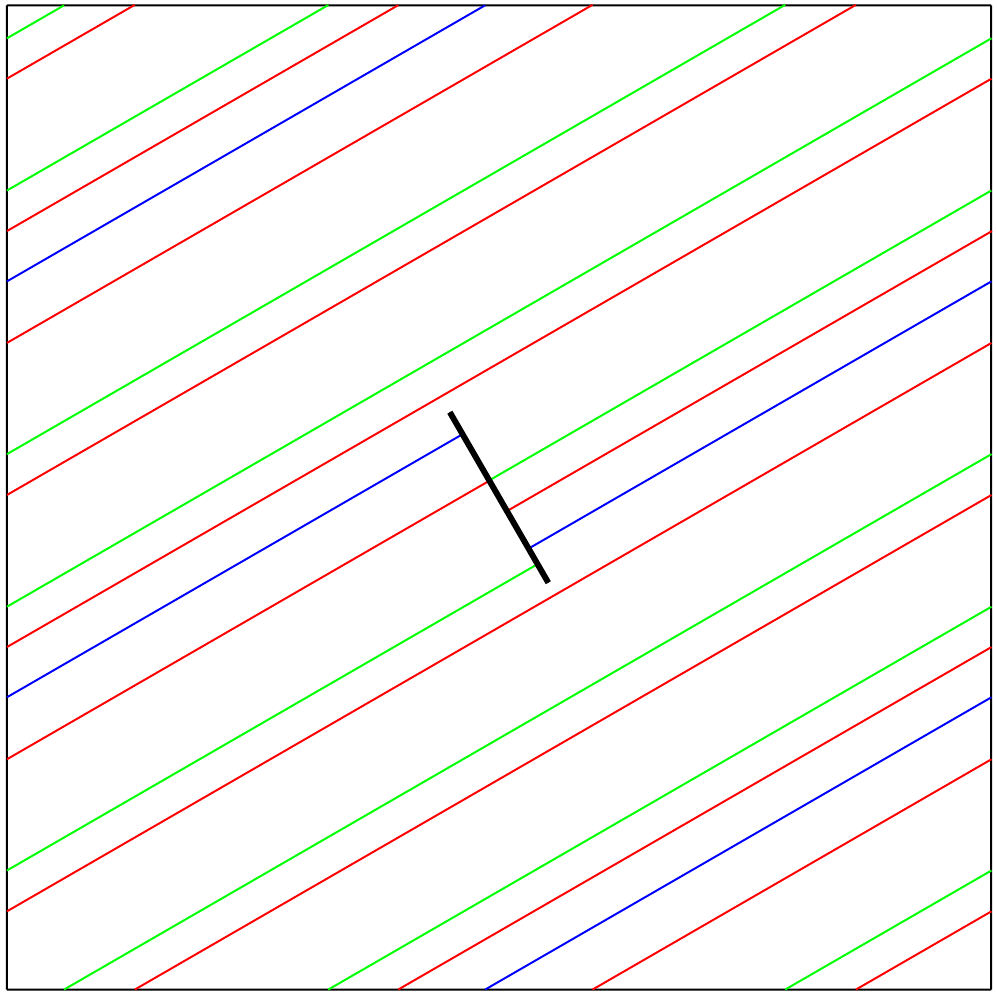}

\caption{Three types of orbits: the blue orbit is the shortest, the red one is the longest, while the
green one is of the intermediate length.The black segment removed is orthogonal to the direction
of the trajectories.}

\end{figure}

Together with E. Caglioti in \cite{CagliotiFG2003}, we proposed the idea of using the Blank-Krikorian 
$3$-length theorem to analyze particle paths in the $2$-dimensional periodic Lorentz gas.

More precisely, orbits with the same lengths in the Blank-Krikorian theorem define a $3$-term partition 
of the flat $2$-torus into parallel strips, whose lengths and widths are computed exactly in terms of the 
continued fraction expansion of the slope (see Figure 17\footnote{Figures 16 and 17 are taken from
a conference by E. Caglioti at the Centre International de Rencontres Math\'ematiques, Marseilles,
February 18-22, 2008.}.)

The collision pattern for particles leaving the surface of one obstacle --- and  therefore the transfer map 
--- can be explicitly determined in this way, for a.e. direction $v\in\bS^1$.

\begin{figure}

\includegraphics[width=10.0cm]{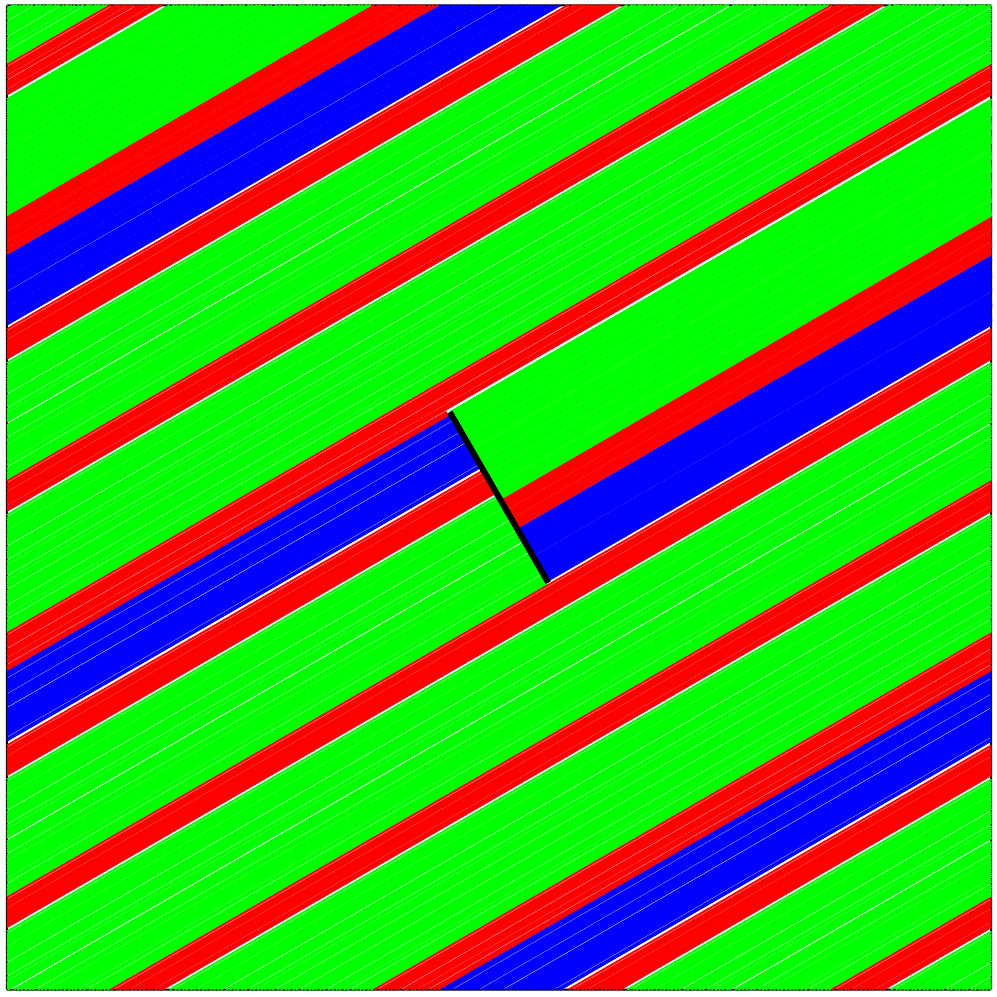}

\caption{The $3$-term partition. The shortest orbits are collected in the blue strip, the longest orbits
in the red strip, while the orbits of intermediate length are collected in the green strip.}

\end{figure}

In fact, there is a classical result known as the $3$-length theorem, which is related to Blank-Krikorian's. 
Whereas the Blank-Krikorian theorem considers a linear flow with irrational slope on the flat $2$-torus, 
the classical $3$-length theorem is a statement about rotations of an irrational angle --- i.e. about
sections of the linear flow with irrational slope.

\begin{Thm}[$3$-length theorem]
Let $\a\in(0,1)\setminus\bQ$ and $N\ge 1$. The sequence 
$$
\{n\a\,|\,0\le n\le N\}
$$
defines $N+1$ intervals on the circle of unit length $\simeq\bR/\bZ$. The lengths of these intervals take 
at most $3$ different values.
\end{Thm}

This striking result was conjectured by H. Steinhaus, and proved in 1957 independently by P. Erd\"os,
G. Hajos, J. Suranyi, N. Swieczkowski, P. Sz\"usz --- reported in \cite{Suranyi1958}, and by Vera S\`os 
\cite{Sos1958}.

\begin{figure}

\includegraphics[width=5.05cm]{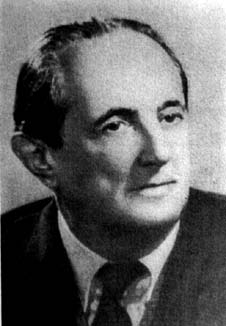}\includegraphics[width=5.60cm]{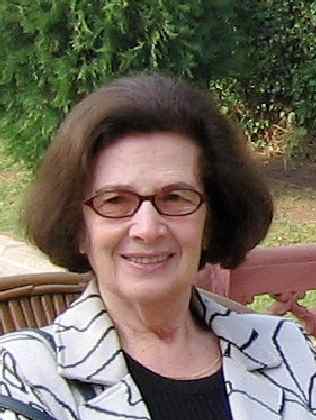}

\caption{Left: Hugo D. Steinhaus (1887-1972); right: Vera T. S\'os}

\end{figure}

\smallskip
As we shall see, the $3$-length theorem (in either form) is the key to encoding particle paths in the 
$2$-dimensional Lorentz gas. We shall need explicitly the formulas giving the lengths and widths 
of the $3$ strips in the partition of the flat $2$-torus defined by the Blank-Krikorian theorem. As this 
is based on the continued fraction expansion of the slope of the linear flow considered in the
Blank-Krikorian theorem, we first recall some basic facts about continued fractions. An excellent
reference for more information on this subject is \cite{Khinchin1964}.

\smallskip
\noindent
\textbf{Continued fractions}

Assume $0<v_2<v_1$ and set $\a=v_2/v_1$, and consider the continued fraction expansion of $\a$: 
$$
\a=[0;a_0,a_1,a_2,\ldots]
	=\frac1{\displaystyle a_0+\frac1{\displaystyle a_1+\ldots}}\,.
$$
Define the sequences of convergents $(p_n,q_n)_{n\ge 0}$ --- meaning that
$$
\frac{p_{n+2}}{q_{n+2}}=[0;a_0,\ldots,a_n]\,,\quad n\ge 2
$$
--- by the recursion formulas 
$$
\begin{array}{ll}
p_{n+1}=a_np_n+p_{n-1}\,,\quad &p_0=1\,,\,\,p_1=0\,,
\\
q_{n+1}=a_nq_n+q_{n-1}\quad &q_0=0\,,\,\,q_1=1\,.
\end{array}
$$
Finally, let $d_n$ denote the sequence of errors
$$
d_n=|q_n\alpha-p_n|=(-1)^{n-1}(q_n\a-p_n)\,,\quad n\ge 0\,,
$$
so that
$$
d_{n+1}=-a_nd_n+d_{n-1}\,,\quad d_0=1\,,\,\,d_1=\a\,.
$$

The sequence $d_n$ is decreasing and converges to $0$, at least exponentially fast. (In fact,
the irrational number for which the rational approximation by continued fractions is the slowest 
is the one for which the sequence of denominators $q_n$ have the slowest growth, i.e. the
golden mean
$$
\th=[0;1,1,\ldots]=\frac1{1+\displaystyle\frac1{1+\ldots}}=\frac{\sqrt{5}-1}2\,.
$$
The sequence of errors associated with $\th$ satisfies $d_{n+1}=-d_n+d_{n-1}$ for each 
$n\ge 1$ with $d_0=1$ and $d_1=\th$, so that $d_n=\th^n$ for each $n\ge 0$.)

By induction, one verifies that
$$
q_nd_{n+1}+q_{n+1}d_n=1\,,\quad n\ge 0\,.
$$
\textbf{Notation:} we write $p_n(\a),q_n(\a),d_n(\a)$ to indicate the dependence of these quantities
in $\a$.

\smallskip
\noindent
\textbf{Collision patterns}

The Blank-Krikorian $3$-length theorem has the following consequence, of fundamental importance 
in our analysis.

Any particle leaving the surface of one obstacle in some irrational direction $v$ will next collide with 
one of \textit{at most $3$} --- exceptionally 2 --- other obstacles.

Any such collision pattern involving the $3$ obstacles seen by the departing particle in the direction of 
its velocity is completely determined by exactly $4$ parameters, computed in terms of the continued 
fraction expansion of $v_2/v_1$ --- in the case where $0<v_2<v_1$, to which the general case can be 
reduced by obvious symmetry arguments.

\begin{figure}

\includegraphics[width=11.0cm]{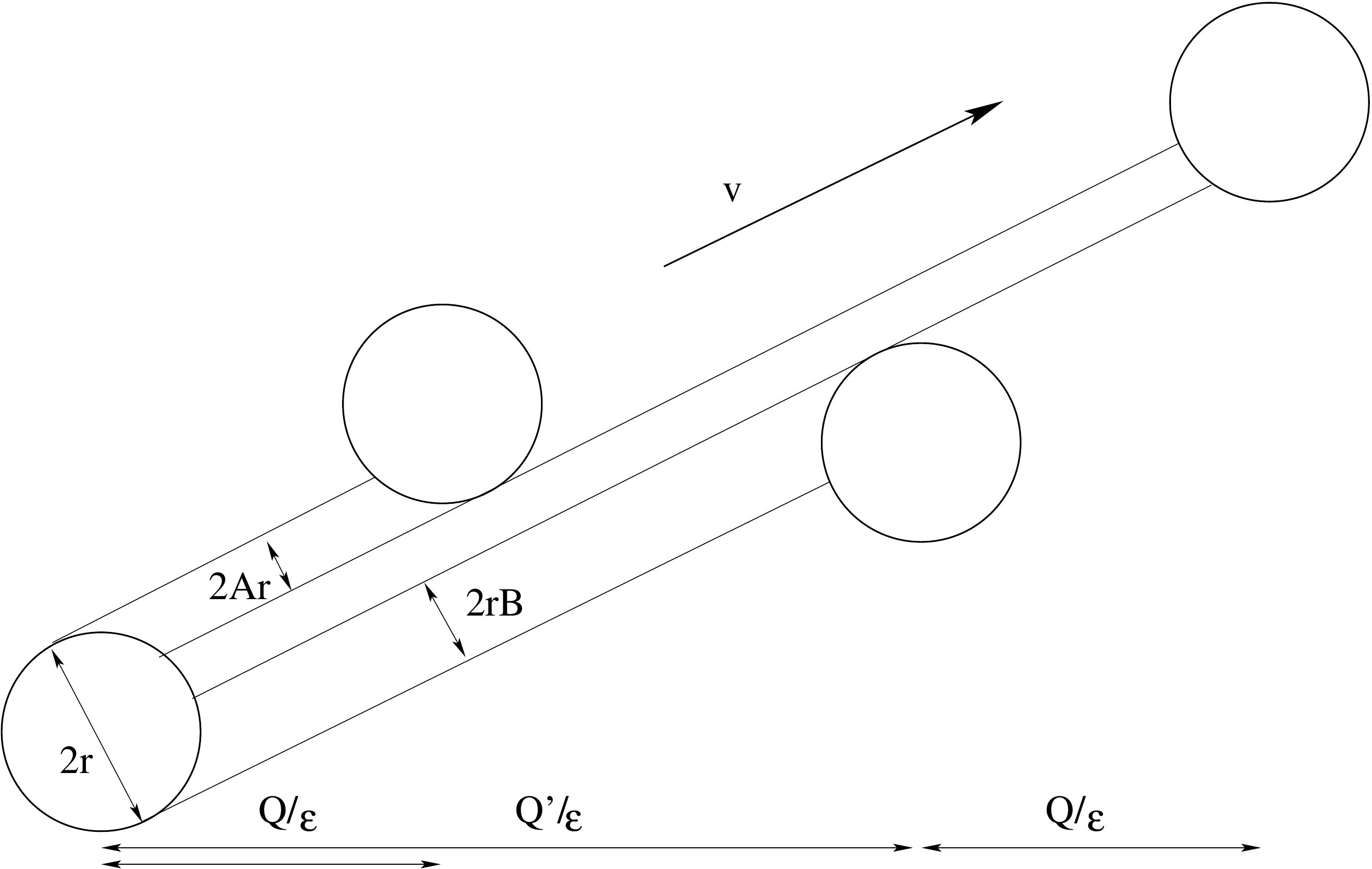}

\caption{Collision pattern seen from the surface of one obstacle. Here, $\eps=2r/v_1$.}

\end{figure}

Assume therefore $0<v_2<v_1$ with $\a=v_2/v_1\notin\bQ$. Henceforth,
we set $\eps=2r\sqrt{1+\a^2}$ and define
$$
\ba
N(\alpha,\eps)&=\inf\{n\ge 0\,|\,d_n(\a)\le\eps\}\,,
\\
k(\alpha,\eps)&=-\left[\frac{\eps-d_{N(\alpha,\eps)-1}(\a)}{d_{N(\alpha,\eps)}(\a)}\right]\,.
\ea
$$
The parameters defining the collision pattern are $A,B,Q$ --- as they appear on the previous figure 
--- together with an extra parameter $\Si\in\{\pm 1\}$. Here is how they are computed in terms of the 
continued fraction expansion of $\a=v_2/v_1$:
$$
\ba
A(v,r)&=1-\tfrac{d_{N(\a,\eps)}(\a)}{\eps}\,,
\\
B(v,r)&=1-\tfrac{d_{N(\a,\eps)-1}(\a)}{\eps}+\tfrac{k(\a,\eps)d_{N(\a,\eps)}(\a)}{\eps}\,,
\\
Q(v,r)&=\eps q_{N(\a,\eps)}(\a)\,,
\\
\Si(v,r)&=(-1)^{N(\a,\eps)}\,.
\ea
$$
The extra-parameter $\Si$ in the list above has the following geometrical meaning. It determines the 
relative position of the closest and next to closest obstacles seen from the particle leaving the surface 
of the obstacle at the origin in the direction $v$.

The case represented on the figure where the closest obstacle is on top of the strip consisting of the 
longest particle path corresponds with $\Si=+1$, the case where that obstacle is at the bottom of this 
same strip corresponds with $\Si=-1$.

The figure above showing one example of collision pattern involves still another parameter, denoted 
$Q'$ on that figure.

This parameter $Q'$ is not independent from $A,B,Q$, since one must have
$$
AQ+BQ'+(1-A-B)(Q+Q')=1
$$
each term in this sum corresponding to the surface of one of the three strips in the 3-term partition of 
the 2-torus. (Remember that the length of the longest orbit in the Blank-Krikorian theorem is the sum 
of the two other lengths.) Therefore
$$
Q'(v,r)=\frac{1-Q(v,r)(1-B(v,r))}{1-A(v,r)}\,.
$$

\smallskip
Once the structure of collision patterns elucidated with the help of the Blank-Krikorian variant of the 
$3$-length theorem, we return to our original problem, namely that of computing the transfer map. 

In the next proposition, we shall see that the transfer map in a given, irrational direction $v\in\bS^1$ 
can be expressed explicitly in terms of the parameters $A,B,Q,\Si$ defining the collision pattern 
correponding with this direction.

Denote
$$
\bK:=]0,1[^3\times\{\pm1\}
$$
and let $(A,B,Q,\Si)\in\bK$ be the parameters defining the collision pattern seen by a particle leaving 
the surface of one obstacle in the direction $v$. Set
$$
\begin{array}{rl}
\mathbf{T}_{A,B,Q,\Si}(h')=&(Q,h'-2\Si(1-A))
\\ \hbox{if }&1-2A<\Si h'\le1\,,
\\
\mathbf{T}_{A,B,Q,\Si}(h')=&\left(Q',h'+2\Si (1-B)\right)
\\ \hbox{if }&-1\le\Si h'<-1+2B\,,
\\
\mathbf{T}_{A,B,Q,\Si }(h')=&\left(Q'+Q,h'+2\Si (A-B)\right)
\\ \hbox{if }&-1+2B\le\Si h'\le1-2A\,.
\end{array}
$$
With this notation, the transfer map is essentially given by the explicit formula $\mathbf{T}_{A,B,Q,\Si}$, 
except for an error of the order $O(r^2)$ on the free-path length from obstacle to obstacle.  

\begin{Prop}[Caglioti-Golse \cite{CagliotiFG2008, CagliotiFG2009}]
One has 
$$
T_r(h',v)=\mathbf{T}_{(A,B,Q,\Si)(v,r)}(h')+(O(r^2),0)
$$
in the limit as $r\to 0^+$.
\end{Prop}

In fact, the proof of this proposition can be read on the figure above that represents a generic collision 
pattern. The first component in the explicit formula
$$
\mathbf{T}_{(A,B,Q,\Si)(v,r)}(h')
$$
represents exactly $r\times$ the distance between the vertical segments that are the projections of the 
diameters of the $4$ obstacles on the vertical ordinate axis. Obviously, the free-path length from 
obstacle to obstacle is the distance between the corresponding vertical segments, minus a quantity 
of the order $O(r)$ that is the distance from the surface of the obstacle to the corresponding vertical 
segment. 

On the other hand, the second component in the same explicit formula is exact, as it relates impact 
parameters, which are precisely the intersections of the infinite line that contains the particle path with 
the vertical segments corresponding with the two obstacles joined by this particle path.

\smallskip
If we summarize what we have done so far, we see that we have solved our first problem stated at the 
beginning of the present section, namely that of finding a convenient way of coding the billiard flow in 
the periodic case and for space dimension $2$, for a.e. given direction $v$.


\section{An ergodic theorem for collision patterns}


It remains to solve the second problem, namely, to find a convenient way of averaging the computation 
above so as to get rid of the dependence on the direction $v$.

Before going further in this direction, we need to recall some known facts about the ergodic theory of 
continued fractions. 

\smallskip
\noindent
\textbf{The Gauss map}

Consider the Gauss map, which is defined on all irrational numbers in $(0,1)$ as follows:
$$
T:\,(0,1)\setminus\bQ\ni x
	\mapsto Tx=\tfrac1x-\left[\tfrac1x\right]\in(0,1)\setminus\bQ\,.
$$

This Gauss map has the following invariant probability measure --- found by Gauss himself:
$$
dg(x)=\tfrac1{\ln 2}\frac{dx}{1+x}\,.
$$

Moreover, the Gauss map $T$ is ergodic for the invariant measure $dg(x)$.  By Birkhoff's theorem, for 
each $f\in L^1(0,1;dg)$
$$
\frac1N\sum_{k=0}^{N-1}f(T^kx)\to\int_0^1f(z)dg(z)\hbox{ a.e. in }x\in(0,1)
$$
as $N\to+\infty$.

How the Gauss map is related to continued fractions is explained as follows: for
$$
\a=[0;a_0,a_1,a_2,\ldots]=\frac1{\displaystyle a_0+\frac1{\displaystyle a_1+\ldots}}
	\in(0,1)\setminus\bQ
$$
the terms $a_k(\a)$ of the continued fraction expansion of $\a$ can be computed from the iterates of 
the Gauss map acting on $\a$: specifically 
$$
a_k(\a)=\left[\frac1{T^k\a}\right]\,,\quad k\ge 0
$$

As a consequence, the Gauss map corresponds with the shift to the left on infinite sequences of positive 
integers arising in the continued fraction expansion of irrationals in $(0,1)$. In other words,
$$
T[0;a_0,a_1,a_2,\ldots]=[0;a_1,a_2,a_3\ldots]\,,
$$
equivalently recast as
$$
a_n(T\a)=a_{n+1}(\a)\,,\quad n\ge 0\,.
$$

Thus, the terms $a_k(\a)$ of the continued fraction expansion of any $\a\in(0,1)\setminus\bQ$ are easily 
expressed in terms of the sequence of iterates $(T^k\a)_{k\ge 0}$ of the Gauss map acting on $\a$. The
error $d_n(\a)$ is also expressed in term of that same sequence $(T^k\a)_{k\ge 0}$, by equally simple
formulas.

Starting from the induction relation on the error terms 
$$
d_{n+1}(\a)=-a_n(\a)d_n(\a)+d_{n-1}(\a)\,,\quad d_0(\a)=1\,,\,\,d_1(\a)=\a
$$
and the explicit formula relating $a_n(T\a)$ to $a_n(\a)$, we see that
$$
\a d_n(T\a)=d_{n+1}(\a)\,,\quad n\ge 0\,.
$$
This entails the formula
$$
d_n(\a)=\prod_{k=0}^{n-1}T^k\a\,,\quad n\ge 0\,.
$$
Observe that, for each $\th\in[0,1]\setminus\bQ$, one has
$$
\th\cdot T\th<\tfrac12\,,
$$
so that
$$
d_n(\a)\le 2^{-[n/2]}\,,\quad n\ge 0\,,
$$
which establishes the exponential decay mentionned above. (As a matter of fact, exponential
convergence is the slowest possible for the continued fraction algorithm, as it corresponds 
with the rational approximation of algebraic numbers of degree $2$, which are the hardest to 
approximate by rational numbers.) 

Unfortunately, the dependence of $q_n(\a)$ in $\a$ is more complicated. Yet one can find a way 
around this, with the following observation. Starting from the relation
$$
q_{n+1}(\a)d_n(\a)+q_n(\a)d_{n+1}(\a)=1\,,
$$
we see that
$$
\ba
q_n(\a)d_{n-1}(\a)&=\sum_{j=1}^n(-1)^{n-j}\frac{d_n(\a)d_{n-1}(\a)}{d_j(\a)d_{j-1}(\a)}	
\\
&=\sum_{j=1}^n(-1)^{n-j}\prod_{k=j}^{n-1}T^{k-1}\a T^k\a	\,.	
\ea
$$
Using once more the inequality $\th\cdot T\th<\tfrac12$ for $\th\in[0,1]\setminus\bQ$, one can
truncate the summation above at the cost of some exponentially small error term. Specifically, 
one finds that
$$
\ba
{}&\left|q_n(\a)d_{n-1}(\a)
	-\sum_{j=n-l}^n(-1)^{n-j}\frac{d_n(\a)d_{n-1}(\a)}{d_j(\a)d_{j-1}(\a)}\right|	
\\
&\qquad=
\left|q_n(\a)d_{n-1}(\a)
	-\sum_{j=n-l}^n(-1)^{n-j}\prod_{k=j}^{n-1}T^{k-1}\a T^k\a\right|\le 2^{-l}\,.
\ea
$$
More information on the ergodic theory of continued fractions can be found in the classical
monograph \cite{Khinchin1964} on continued fractions, and in Sinai's book on ergodic
theory \cite{Sinai1994}.

\smallskip
\noindent
\textbf{An ergodic theorem}

We have seen in the previous section that the transfer map satisfies
$$
T_r(h',v)=\mathbf{T}_{(A,B,Q,\Si)(v,r)}(h')+(O(r^2),0)\hbox{ as }r\to 0^+
$$
for each $v\in\bS^1$ such that $v_2/v_1\in(0,1)\setminus\bQ$.

Obviously, the parameters $(A,B,Q,\Si)$ are extremely sensitive to variations in $v$ and $r$ as 
$r\to 0^+$, so that even the explicit formula for $T_{A,B,Q,\Si}$, is not too useful in itself. 

Each time one must handle a strongly oscillating quantity such as the free path length $\tau_r(x,v)$ 
or the transfer map $T_r(h',v)$, it is usually a good idea to consider the distribution of that quantity 
under some natural probability measure than the quantity itself. Following this principle, we are led 
to consider the family of probability measures in $(s,h)\in\bR_+\times[-1,1]$ 
$$
\de((s,h)-T_r(h',v))\,,
$$
or equivalently
$$
\de((s,h)-T_{(A,B,Q,\Si)(v,r)}(h'))\,.
$$

A first obvious idea would be to average out the dependence in $v$ of this family of measures: as we 
shall see later, this is not an easy task.

A somewhat less obvious idea is to average over obstacle radius. Perhaps surprisingly, this is easier 
than averaging over the direction $v$.

That averaging over obstacle radius is a natural operation in this context can be explained by the 
following observation. We recall that the sequence of errors $d_n(\a)$ in the continued fraction 
expansion of an irrational $\a\in(0,1)$ satisfies
$$
\a d_n(T\a)=d_{n+1}(\a)\,,\quad n\ge 0\,,
$$
so that 
$$
N(\a,\eps)=\inf\{n\ge 1\,|\,d_n(\a)\le\eps\}
$$
is transformed by the Gauss map as follows:
$$
N(a,\eps)=N(T\a,\eps/\a)+1\,.
$$

In other words, the transfer map for the $2$-dimensional periodic Lorentz gas in the billiard table 
$Z_r$ (meaning, with circular obstacles of radius $r$ centered at the vertices of the lattice $\bZ^2$) 
in the direction $v$ corresponding with the slope $\a$ is essentially the same as for the billiard table 
$Z_{r/\a}$ but in the direction corresponding with the slope $T\a$. Since the problem is invariant 
under the transformation
$$
\a\mapsto T\a\,,\qquad r\mapsto r/\a
$$
this suggests the idea of averaging with respect to the scale invariant measure in the variable $r$, 
i.e. $\frac{dr}{r}$ on $\bR_+^*$.

\smallskip
The key result in this direction is the following ergodic lemma for functions that depend on 
\textit{finitely many} $d_n$s.

\begin{Lem}[Caglioti-Golse \cite{CagliotiFG2003, Golse2006, CagliotiFG2009}]
For $\a\in(0,1)\setminus\bQ$, set
$$
N(\a,\eps)=\inf\{n\ge 0\,|\,d_n(\a)\le\eps\}\,.
$$
For each $m\ge 0$ and each $f\in C(\bR_+^{m+1})$, one has 
$$
\frac1{|\ln\eta|}\int_\eta^{1/4}f\left(\frac{d_{N(\a,\eps)}(\a)}{\eps},\ldots,
	\frac{d_{N(\a,\eps)-m}(\a)}{\eps}\right)\frac{d\eps}\eps\to L_m(f)
$$
a.e. in $\a\in(0,1)$ as $\eta\to 0^+$, where the limit $L_m(f)$ is independent of $\a$.
\end{Lem}

With this lemma, we can average over obstacle radius any function that depends on collision patterns, 
i.e. any function of the parameters $A,B,Q,\Si$.

\begin{Prop}[Caglioti-Golse \cite{CagliotiFG2009}]
\label{T-ErgoABQsi}
Let $\bK=[0,1]^3\times\{\pm1\}$. For each $F\in C(\bK)$, there exists $\cL(F)\in\bR$  independent of 
$v$ such that
$$
\frac1{\ln(1/\eta)}\int_\eta^{1/2}F(A(v,r),B(v,r),Q(v,r),\Si(v,r))\frac{dr}{r}\to\cL(F)
$$
for a.e. $v\in\mathbf{S}^1$ such that $0<v_2<v_1$ in the limit as $\eta\to 0^+$.
\end{Prop}

\begin{proof}[Sketch of the proof] First eliminate the $\Si$ dependence by decomposing
$$
F(A,B,Q,\Si)=F_+(A,B,Q)+\Si F_-(A,B,Q)\,.
$$
Hence it suffices to consider the case where $F\equiv F(A,B,Q)$.

Setting $\a=v_2/v_1$ and $\eps=2r/v_1$, we recall that
$$
\ba
A(v,r)&\hbox{ is a function of }\frac{d_{N(\a,\eps)}(\a)}{\eps}\,,
\\
B(v,r)&\hbox{ is a function of }\frac{d_{N(\a,\eps)}(\a)}{\eps}
	\hbox{ and }\frac{d_{N(\a,\eps)-1}(\a)}{\eps}\,.
\ea
$$

As for the dependence of $F$ on $Q$, proceed as follows: in $F(A,B,Q)$, replace $Q(v,r)$ with
$$
\frac{\eps}{d_{N(\a,\eps)-1}}
\sum_{j=N(\a,\eps)-l}^{N(\a,\eps)}(-1)^{{N(\a,\eps)}-j}
	\frac{d_{N(\a,\eps)}(\a)d_{{N(\a,\eps)}-1}(\a)}{d_j(\a)d_{j-1}(\a)}\,,
$$
at the expense of an error term of the order 
$$
O(\hbox{modulus of continuity of $F$}(2^{-m}))\to 0\hbox{ as }l\to 0\,,
$$
uniformly as $\eps\to 0^+$.

This substitution leads to an integrand of the form
$$
f\left(\frac{d_{N(\a,\eps)}(\a)}{\eps},
	\ldots,\frac{d_{N(\a,\eps)-m-1}(\a)}{\eps}\right)
$$
to which we apply the ergodic lemma above: its Ces\`aro mean converges, in the small radius limit, 
to some limit $\cL_m(F)$ independent of $\a$.

By uniform continuity of $F$, one finds that 
$$
|\cL_m(F)-\cL_{m'}(F)|=O(\hbox{modulus of continuity of $F$}(2^{-m\vee m'}))
$$
(with the notation $m\vee m'=\max(v,v')$), so that $\cL_m(F)$ is a Cauchy sequence as $m\to\infty$. 
Hence 
$$
\cL_m(F)\to\cL(F)\hbox{ as }m\to\infty
$$
and with the error estimate above for the integrand, one finds that 
$$
\frac1{\ln(1/\eta)}\int_\eta^{1/2}F(A(v,r),B(v,r),Q(v,r),\Si(v,r))\frac{dr}{r}\to\cL(F)
$$
as $\eta\to 0^+$.
\end{proof}

\smallskip
With the ergodic theorem above, and the explicit approximation of the transfer map expressed in 
terms of the parameters $(A,B,Q,\Si)$ that determine collision patterns in any given direction $v$, 
we easily arrive at the following notion of a ``probability of transition" for a particle leaving the surface 
of an obstacle with an impact parameter $h'$ to hit the next obstacle on its trajectory at time $s/r$ with 
an impact parameter $h$.

\begin{Thm}[Caglioti-Golse, \cite{CagliotiFG2008, CagliotiFG2009}]
For each $h'\in[-1,1]$, there exists a probability density $P(s,h|h')$ on $\bR_+\times[-1,1]$ such that, 
for each $f\in C(\bR_+\times[-1,1])$,
$$
\frac1{|\ln\eta|}\int_\eta^{1/4}f(T_r(h',v))\frac{dr}{r}\to\int_0^\infty\int_{-1}^1f(s,h)P(s,h|h')dsdh
$$
a.e. in $v\in\bS^1$ as $\eta\to 0^+$.
\end{Thm}

In other words, the transfer map converges in distribution and in the sense of Ces\`aro, in the small 
radius limit, to a transition probability $P(s,h|h')$ that is independent of $v$.

\smallskip
We are therefore left with the following problems:

\noindent
a) to compute the transition probability $P(s,h|h')$ explicitly and discuss its properties, and

\noindent
b) to explain the role of this transition probability in the Boltzmann-Grad limit of the periodic Lorentz 
gas dynamics.


\section{Explicit computation \\ of the transition probability $P(s,h|h')$}


Most unfortunately, our argument leading to the existence of the limit $\cL(F)$, the core result of the 
previous section, cannot be used for computing explicitly the value $\cL(F)$. Indeed, the convergence 
proof is based on the ergodic lemma in the last section, coupled to a sequence of approximations of 
the parameter $Q$ in collision patterns that involve only finitely many error terms $d_n(\a)$ in the 
continued fraction expansion of $\a$. The existence of the limit is obtained through Cauchy's criterion, 
precisely because of the difficulty in finding an explicit expression for the limit.

Nevertheless, we have arrived at the following expression for the transition
probability $P(s,h|h')$:

\begin{Thm}[Caglioti-Golse \cite{CagliotiFG2008, CagliotiFG2009}]\label{T-TransiProba}
The transition probability density $P(s,h|h')$ is expressed in terms of $a=\tfrac12|h-h'|$  and 
$b=\tfrac12|h+h'|$ by the explicit formula
$$
\ba
P(s,h|h')=\frac{3}{\pi^2sa}\Big[\left((s-\tfrac12sa)\!\wedge\!(1+\tfrac12sa)
	\!-\!(1\!\vee\!(\tfrac12s+\tfrac12sb)\right)_+
\\
+\!\left((s-\textstyle\frac12sa)\!\wedge\! 1\!
	-\!((\tfrac12s+\tfrac12sb)\!\vee\!\left(1\!-\!\textstyle\frac12sa\right)\right)_+
\\
\!+\!sa\!\wedge\!|1-s|\mathbf{1}_{s<1}+(sa\!-\!|1-s|)_+\Big]\,,
\ea
$$
with the notations $x\wedge y=\min(x,y)$ and $x\vee y=\max(x,y)$.

Moreover, the function
$$
(s,h,h')\mapsto (1+s)P(s,h|h')\hbox{ belongs to }L^2(\bR_+\times[-1,1]^2)\,.
$$
\end{Thm} 

In fact, the key result in the proof of this theorem is the asymptotic distribution of $3$-obstacle 
collision patterns --- i.e. the computation of the limit $\cL(f)$, whose existence has been proved 
in the last section's proposition.

\begin{Thm}[Caglioti-Golse \cite{CagliotiFG2009}]\label{T-Meas-m}
Define $\bK=[0,1]^3\times\{\pm1\}$; then, for each $F\in C(\bK)$
$$
\ba
\frac1{|\ln\eta|}\int_\eta^{1/4}F((A,B,Q,\Si)(v,r))\frac{dr}{r}
	\to\cL(F)
\\
=\int_{\bK}F(A,B,Q,\Si)dm(A,B,Q,\Si)\hbox{ a.e. in }v\in\bS^1
\ea
$$
as $\eta\to 0^+$, where
$$
\ba
dm(A,B,Q,\Si)&=dm_0(A,B,Q)\otimes\tfrac12(\de_{\Si=1}+\de_{\Si=-1})\,,
\\
dm_0(A,B,Q)&=\tfrac{12}{\pi^2}
\indc_{0<A<1}\indc_{0<B<1-A}\indc_{0<Q<\frac1{2-A-B}}	\frac{dAdBdQ}{1-A}\,.
\ea
$$
\end{Thm} 

\smallskip
Before giving an idea of the proof of the theorem above on the distribution of $3$-obstacle collision 
patterns, it is perhaps worthwhile explaining why the measure $m$ above is somehow natural in the 
present context.

To begin with, the constraints $0<A<1$ and $0<B<1-A$ have an obvious geometric meaning (see 
figure 18 on collision patterns.) More precisely, the widths of the three strips in the $3$-term partition 
of the $2$-torus minus the slit constructed in the penultimate section (as a consequence of the
Blank-Krikorian $3$-length theorem) add up to $1$. Since $A$ is the width of the strip consisting of 
the shortest orbits in the Blank-Krikorian theorem, and $B$ that of the strip consisting of the next to 
shortest orbits, one has
$$
0<A+B\le 1
$$
with equality only in the exceptional case where the orbits have at most $2$ different lengths, which 
occurs for a set of measure $0$ in $v$ or $r$. Therefore, one has
$$
0<B(v,r)<1-A(v,r)\,,\quad\hbox{ for a.e. }r\in(0,\tfrac12)\,.
$$

Likewise, the total area of the $2$-torus is the sum of the areas of the strips consisting of all orbits with 
the $3$ possible lengths:
$$
\ba
1=QA+Q'B+(Q+Q')(1-A-B)&=Q(1-B)+Q'(1-A)
\\
&\ge Q(2-A-B)
\ea
$$
as $Q'\ge Q$ (see again the figure above on collision patterns.)

Therefore, the volume element 
$$
\frac{dAdBdQ}{1-A}
$$
in the expression of $dm_0$ imples that the parameters $A$, $\frac{B}{1-A}$ --- or equivalently 
$B$ mod. $1-A$ --- and $Q$ are independent and uniformly distributed in the largest subdomain 
of $[0,1]^3$  that is compatible with the geometric constraints.

\smallskip
The first theorem is a consequence of the second: indeed, $P(s,h|h')dsdh$ is the image measure of 
$dm(A,B,Q,\Si)$ under the map
$$
\bK\ni(A,B,Q,\Si)\mapsto T_{(A,B,Q,\Si)}(h',v)\,.
$$
That $(1+s)P(s,h|h')$ is square integrable is proved by inspection --- by using the explicit formula for 
$P(s,h|h')$.

\smallskip
Therefore, it remains to prove the second theorem. 

We are first going to show that the family of averages over velocities satisfy
$$
\ba
\int_{|v|=1\atop 0<v_2<v_1}&F(A(v,r),B(v,r),Q(v,r),\Si(v,r))dv
\\
&\to\tfrac{\pi}8\int_{\bK}F(A,B,Q,\Si)dm(A,B,Q,\Si)
\ea
$$
as $r\to 0^+$ for each $F\in C_b(\bK)$.

On the other hand, because of the proposition in the previous section
$$
\frac1{\ln(1/\eta)}\int_\eta^{1/2}F(A(v,r),B(v,r),Q(v,r),\Si(v,r))\frac{dr}{r}\to\cL(F)
$$
for a.e. $v\in\mathbf{S}^1$ such that $0<v_2<v_1$ in the limit as $\eta\to 0^+$. 

Since we know that the limit $\cL(F)$ is independent of $v$, comparing 
the two convergence statements above shows that
$$
\cL(F)=\int_{\bK}F(A,B,Q,\Si)dm(A,B,Q,\Si)\,.
$$

Therefore, we are left with the task of computing
$$
\lim_{r\to 0^+}\int_{|v|=1\atop 0<v_2<v_1}F(A(v,r),B(v,r),Q(v,r),\Si(v,r))dv\,.
$$
The method for computing this type of expression is based on

a) Farey fractions (sometimes called ``slow continued fractions"), and

b) estimates for Kloosterman's sums, due to Boca-Zaharescu \cite{BocaZaha2007}.

To begin with, we need to recall a few basic facts about Farey fractions.

\smallskip
\noindent
\textbf{Farey fractions}

Put a filtration on the set of rationals in $[0,1]$ as follows
$$
\cF_\cQ=\{\tfrac{p}{q}\,|\,0\le p\le q\le \cQ\,,\,\,\hbox{g.c.d.}(p,q)=1\}
$$
indexed in increasing order:
$$
0=\frac01<\g_1<\ldots
	<\g_j=\frac{p_j}{q_j}<\ldots<\g_{\varphi(\cQ)}=\frac11=1
$$
where $\varphi$ denotes Euler's totient function:
$$
\phi(n)=n\prod_{p\hbox{ \tiny{prime} }\atop p|n}\left(1-\frac1p\right)\,.
$$

An important operation in the construction of Farey fractions is the
notion of ``mediant" of two fractions. Given two rationals
$$
\g=\frac{p}q\hbox{ and }\hat\g=\frac{\hat p}{\hat q}
$$ 
with 
$$
0\le p\le q\,,\quad 0\le\hat p\le\hat q\,,\hbox{ and }
	\hbox{ g.c.d.}(p,q)=\hbox{ g.c.d.}(\hat p,\hat q)=1\,,
$$
their mediant is defined as
$$
\hbox{mediant}
	=\g\oplus\hat\g:=\frac{p+\hat p}{q+\hat q}\in(\g,\hat\g)\,.
$$

\begin{figure}

\includegraphics{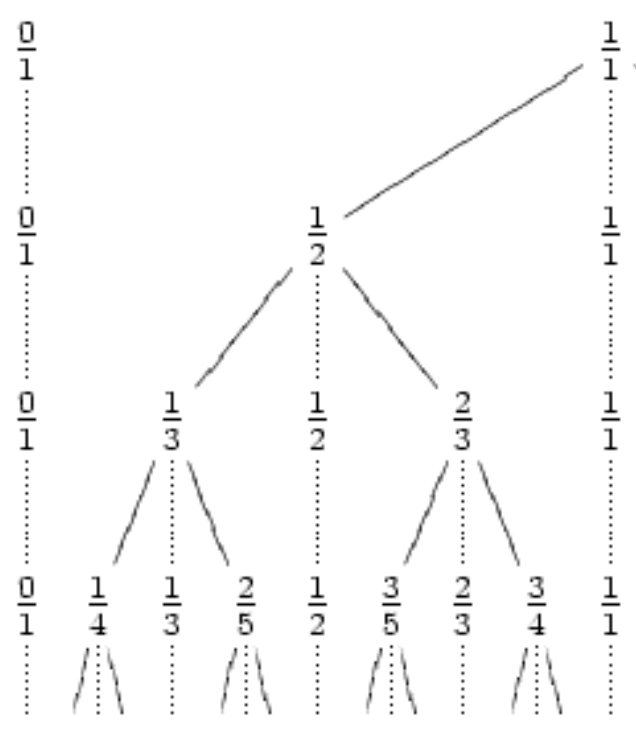}

\caption{The Stern-Brocot tree. Each fraction $\g$ on the $n$-th line is the mediant of the two fractions 
closest to $\g$ on the $n-1$-st line. The first line consists of $0$ and $1$ written as $0=\frac01$ and
$1=\frac11$. Each rational in $[0,1]$ is obtained in this way.}

\end{figure}

Hence, if $\g=\frac{p}q<\hat\g=\frac{\hat p}{\hat q}$ adjacent in $\cF_\cQ$, then
$$
\hat aq-a\hat q=1\hbox{ and }q+\hat q>\cQ\,.
$$
Conversely, $q,\hat q$ are denominators of adjacent fractions in $\cF_\cQ$ if and only if
$$
0\le q,\hat q\le\cQ\,,\quad q+\hat q>\cQ\,,\quad\hbox{g.c.d.}(q,q')=1\,.
$$

Given $\a\in(0,1)\setminus\bQ$ and $\cQ\ge 1$, there exists a unique pair of adjacent Farey fractions 
in $\cF_\cQ$, henceforth denoted $\g(\a,\cQ)$ and $\hat\g(\a,\cQ)$, such that
$$
\g(\a,\cQ)=\frac{p(\a,\cQ)}{q(\a,\cQ)}<\a<\hat\g(\a,\cQ)
	=\frac{\hat p(\a,\cQ)}{\hat q(\a,\cQ)}\,.
$$

At this point, we recall the relation between Farey and continued fractions.

Pick $0<\eps<1$; we recall that, for each $\a\in(0,1)\setminus\bQ$,
$$
N(\a,\eps)=\min\{n\in\bN\,|\,d_n(\a)\le\eps\}\,,\quad d_n(\a)=\hbox{dist}(q_n(\a)\a,\bZ)\,.
$$
Set $\cQ=[1/\eps]$, and let 
$$
\g(\a,\cQ)=\frac{p(\a,\cQ))}{q(\a,\cQ)}<\hat\g(\a,\cQ)=\frac{\hat p(\a,\cQ))}{\hat q(\a,\cQ)}
$$ 
with $g.c.d.(p(\a,\cQ)),q(\a,\cQ))=g.c.d.(\hat p(\a,\cQ)),\hat q(\a,\cQ))=1$ be the two adjacent Farey 
fractions in $\cF_\cQ$ surrounding $\a$. Then

\smallskip
\noindent
a) one of the two integers $q(\a,\cQ)$ and ${\hat q(\a,\cQ)}$ is the denominator of the $N(\a,\eps)$-th 
convergent in the continued fraction expansion of $\a$, i.e. $q_{N(\a,\eps)}(\a)$, and

\noindent
b) the other is of the form
$$
mq_{N(\a,\eps)}+q_{N(\a,\eps)-1}\,,\quad\hbox{ with }0\le m\le a_{N(\a,\eps)}(\a)\,,
$$
where we recall that 
$$
\a=[0;a_1,a_2,\ldots]=\frac1{a_0+\displaystyle\frac1{a_1+\ldots}}\,.
$$

Setting $\a=v_2/v_1$ and $\eps=2r/v_1$, we recall that, by definition
$$
Q(v,r)=\eps q_{N(\a,\eps)}(\a)\in\{\eps q(\a,\cQ),\eps\hat q(\a,\cQ)\}\hbox{ with }\cQ=[1/\eps]\,,
$$
and we further define
$$
D(v,r)=d_{N(\a,\eps)}/\eps=\hbox{dist}(\tfrac1\eps Q(v,r)\a,\bZ)/\eps\,,
$$
and
$$
\ba
\tilde Q(v,r)=\eps\hat q(\a,\cQ)\hbox{ if }q_{N(\a,\eps)}(\a)=q(\a,\cQ)\,,
\\
\tilde Q(v,r)=\eps q(\a,\cQ)\hbox{ if }q_{N(\a,\eps)}(\a)=\hat q(\a,\cQ)\,.
\ea
$$
Now, we recall that $A(v,r)=1-D(v,r)$; moreover, we see that
$$
\ba
B(v,r)&=1-\frac{d_{N(\a,\eps)-1}(\a)}{\eps}
	-\left[\frac{1-d_{N(\a,\eps)-1}(\a)/\eps}{D(v,r)}\right]D(v,r)
\\
\\
&=1-d_{N(\a,\eps)-1}(\a)/\eps\hbox{ mod. }D(v,r)
\\
\\
&=1-\hbox{dist}(\tfrac1\eps\tilde Q(v,r)\a,\bZ)/\eps\hbox{ mod. }D(v,r)\,.
\ea
$$ 

To summarize, we have
$$
F(A(v,r),B(v,r),Q(v,r))=G(Q(v,r),\tilde Q(v,r),D(v,r))
$$
and we are left with the task of computing
$$
\lim_{r\to 0^+}\int_{\bS_+^1}G(Q(v,r),\tilde Q(v,r),D(v,r))dv
$$
where $\bS^1_+$ is the first octant in the unit circle. The other octants in the unit circle give the 
same contribution by obvious symmetry arguments.

More specifically:

\begin{Lem} 
Let $\a\in(0,1)\setminus\bQ$, and let $\tfrac{p}{q}<\a<\tfrac{\hat p}{\hat q}$ be the two adjacent 
Farey fractions in $\cF_\cQ$ surrounding $\a$, with $\cQ=[1/\eps]$. Then

\noindent
a) if $\frac{p}{q}<\a\le\frac{\hat p-\eps}{\hat q}$, then
$$
Q(v,r)=\eps q\,,\quad\tilde Q(v,r)=\eps\hat q\,,\quad D(v,r)=\tfrac1\eps(\a q-p)\,;
$$
b) if $\frac{p+\eps}{q}<\a<\frac{\hat p}{\hat q}$, then
$$
Q(v,r)=\eps\hat q\,,\quad\tilde Q(v,r)=\eps q\,,
		\quad D(v,r)=\tfrac1\eps(\hat p-\a\hat q)\,;
$$
c) if $\frac{p+\eps}{q}<\a\le\frac{\hat p-\eps}{\hat q}$, then
$$
Q(v,r)=\eps(q\wedge\hat q)\,,\quad\tilde Q(v,r)=\eps(q\vee\hat q)\,,
	\quad D(v,r)=\hbox{dist}(\tfrac1\eps Q(v,r)\a,\bZ)\,.
$$
\end{Lem}

Therefore, assuming for simplicity that
$$
G(x,y,z)=g(x,y)H'(z)\hbox{ and }\eps=1/\cQ\,,
$$
one has
$$
\ba
\int_{\bS_+^1}G(Q(v,r),\hat Q(v,r),D(v,r))dv
\\
=
\sum_{0<q,\hat q\le\cQ<q+\hat q\atop g.c.d.(q,\hat q)=1}
\int_{p/q}^{(\hat p-\eps)/\hat q} g\left(\frac{q}{\cQ},\frac{\hat q}{\cQ}\right)
H'(\cQ(q\a-p))d\a
\\
+\hbox{ three other similar terms }
\\
=\sum_{0<q,\hat q\le\cQ<q+\hat q\atop g.c.d.(q,\hat q)=1}
g\left(\frac{q}{\cQ},\frac{\hat q}{\cQ}\right)
	\frac1{q\cQ}\left(H\left(\frac{1-q/\cQ}{\hat q/\cQ}\right)-H(0)\right)
\\
+\hbox{ three other similar terms.}
\ea
$$
Then, everything reduces to computing 
$$
\lim_{\cQ\to+\infty}\frac1{\cQ^2}\sum_{0<q,\hat q\le\cQ<q+\hat q\atop g.c.d.(q,\hat q)=1}
	\psi\left(\frac{q}{\cQ},\frac{\hat q}{\cQ}\right)\,.
$$
We conclude with the following

\begin{Lem}[Boca-Zaharescu \cite{BocaZaha2007}]
For $\psi\in C_c(\bR^2)$, one has
$$
\frac1{\cQ^2}\sum_{0<q,\hat q\le\cQ<q+\hat q\atop g.c.d.(q,\hat q)=1}
	\psi\left(\frac{q}{\cQ},\frac{\hat q}{\cQ}\right)\to
\tfrac{6}{\pi^2}\iint_{0<x,y<1<x+y}\psi(x,y)dxdy
$$
in the limit as $\cQ\to \infty$.
\end{Lem}

This is precisely the path followed by F. Boca and A. Zaharescu to compute the limiting distribution
of free path lengths in \cite{BocaZaha2007} (see Theorem \ref{T-BocaZaha}); as explained above, 
their analysis can be greatly generalized in order to compute the transition probability that is the 
limit of the transfer map as the obstacle radius $r\to 0^+$.


\section{A kinetic theory in extended phase-space \\ for the Boltzmann-Grad limit \\
of the periodic Lorentz gas}


We are now ready to propose an equation for the Boltzmann-Grad limit of the periodic Lorentz gas
in space dimension 2. For each $r\in]0,\frac12[$, denote 
$$
\mathbf{B}_r:\,\Gamma^+_r\ni(x,v)\mapsto\mathbf{B}_r(x,v)
	=(x+\tau_r(x,v)v,\mathcal{R}[x+\tau_r(x,v)v]v)\in\Gamma^+_r
$$
the billiard map. For $(x_0,v_0)\in\Gamma^+_r$, set 
$$
(x_n,v_n)=\mathbf{B}^n_r(x_0,v_0)
$$
and define
$$
b^n_r(x,v)=(A,B,Q,\Si)(v_n,r)\,,\quad n\in\bN^*\,.
$$
Henceforth, for each $n\ge 1$, we denote 
$$
\cK_n:=\bR^2\times\bS^1\times\bR_+\times[-1,1]\times\bK^n\,.
$$

\smallskip
We make the following asymptotic independence hypothesis: there exists a probability measure 
$\Pi$ on $\bR_+\times[-1,1]$ such that, for each $n\ge 1$ and each $\Psi\in C(\cK_n)$ with compact 
support
$$
\ba
{}&\lim_{r\to 0^+}
\int_{Z_r\times\bS^1}\Psi(x,v,r\tau_r({\tfrac{x}r},v),
	h_r({\tfrac{x_1}r},v_1),b^1_r,\ldots,b^n_r)dxdv
\\
&=
\int_{Q_n}
\Psi(x,v,\tau,h,\b_1,\ldots,\b_n)dxdvd\Pi(\tau,h)dm(\b_1)\ldots dm(\b_n)\,,
\ea
\leqno(H)
$$
where 
$$
(x_0,v_0)=(x-\tau_r(x,-v)v,v)\,,\hbox{ and }h_r({\tfrac{x_1}r},v_1)=\sin(n_{x_1},v_1)
$$
and $m$ is the probability measure on $\bK$ obtained in Theorem \ref{T-Meas-m}.

If this holds, the iterates of the transfer map $T_r$ are described by the Markov chain with transition 
probability $P(s,h|h')$.  

This leads to a kinetic equation on an extended phase space for the Boltzmann-Grad limit of the 
periodic Lorentz gas in space dimension 2:
$$
\ba
{}&F(t,x,v,s,h)=
\\
&\qquad\hbox{density of particles with velocity $v$ and position $x$ at time $t$}
\\ 
&\qquad\hbox{that will hit an obstacle after time $s$, with impact parameter $h$.}
\ea
$$

\begin{Thm}[Caglioti-Golse \cite{CagliotiFG2008, CagliotiFG2009}]\label{T-LimitEq}
Assume (H), and let $f^{in}\ge 0$ belong to $C_c(\bR^2\times\bS^1)$. Then one has
$$
f_r\to\int_0^\infty\int_{-1}^1F(\cdot,\cdot,\cdot,s,h)dsdh
	\hbox{ in $L^\infty(\mathbf{R}_+\times\mathbf{R}^2\times\mathbf{S}^1)$ 
		weak-$*$}
$$
in the limit as $r\to 0^+$, where $F\equiv F(t,x,v,s,h)$ is the solution of
$$
\ba
(\partial_t+&v\cdot\nabla_x-\partial_s)F(t,x,v,s,h)
\\
&=\int_{-1}^1P(s,h|h')F(t,x,R[\pi-2\arcsin(h')]v,0,h')dh'\,,
\\
&\qquad F(0,x,v,s,h)=f^{in}(x,v)\int_s^\infty\int_{-1}^1P(\tau,h|h')dh'd\tau\,,
\ea
$$
with $(x,v,s,h)$ running through $\bR^2\times\bS^1\times\bR^*_+\times]-1,1[$. The notation $R[\th]$ 
designates the rotation of an angle $\th$.
\end{Thm}

\smallskip
Let us briefly sketch the computation leading to the kinetic equation above in the extended phase 
space $\cZ=\bR^2\times\bS^1\times\bR_+\times[-1,1]$.

In the limit as $r\to 0^+$, the sequence $(b^{n}_r(x,v))_{n\ge 1}$ converges to a sequence of i.i.d.
random variables with values in $\bK=[0,1]\times\{\pm 1\}$, according to assumption (H).

Then, for each $s_0>0$ and $h_0\in[-1,1]$, we construct a Markov chain $(s_n,h_n)_{n\ge 1}$ with 
values in $\bR_+\times[-1,1]$ in the following manner:
$$
(s_n,h_n)=\bT_{b_n}(h_{n-1})\,,\quad n\ge 1\,.
$$

Now we define the jump process $(X_t,V_t,S_t,H_t)$ starting from $(x,v,s,h)$ in the following manner. 

First pick a trajectory of the sequence $\bbb=(b_n)_{n\ge 1}$;  then, for each $s>0$ and each 
$h\in[-1,1]$, set
$$
(s_0,h_0)=(s,h)\,.
$$
Define then inductively $s_n$ and $h_n$ for $n\ge 1$ by the formula above, together with
$$
\si_n=s_0+\ldots+s_{n-1}\,,\quad n\ge 1\,,
$$
and
$$
v_n=R[2\arcsin(h_{n-1})-\pi]v_{n-1}\,,\quad n\ge 1\,.
$$

With the sequence $(v_n,s_n,h_n)_{n\ge 1}$ so defined, we next introduce the formulas for 
$(X_t,V_t,S_t,H_t)$:

\begin{itemize}

\item While $0\le t<\tau$, we set
$$
\ba
X_t(x,v,s,h)&=x+t\om\,,\qquad &&S_t(x,v,s,h)&&=s-t\,,
\\
V_t(x,v,s,h)&=v\,,\qquad &&H_t(x,v,s,h)&&=h\,.
\ea
$$

\item For $\si_n<t<\si_{n+1}$, we set
$$
\ba
X_t(x,v,s,h)&=x+(t-\si_n)v_n\,,
\\
V_t(x,v,s,h)&=v_n\,,
\\
T_t(x,v,s,h)&=\si_{n+1}-t\,,
\\
H_t(x,v,s,h)&=h_n\,.
\ea
$$
\end{itemize}

To summarize, the prescription above defines, for each $t\ge 0$, a map denoted $T_t$:
$$
\cZ\times\bK^{\bN^*}\ni(x,v,s,h,\bbb)\mapsto T_t(x,\om,\tau,h)=(X_t,V_t,S_t,H_t)\in Z
$$
that is piecewise continuous in $t\in\bR_+$.

Denote by $f^{in}\equiv f^{in}(x,v,s,h)$ the initial distribution function in the extended phase space 
$\cZ$, and by $\chi\equiv\chi(x,v,s,h)$ an observable --- without loss of generality, we assume that 
$\chi\in C^\infty_c(Z)$.

Define $f(t,\cdot,\cdot,\cdot,\cdot)$ by the formula
$$
\begin{aligned}
\int\!\!\!\iiint_{Z} 
	&\chi(x,v,s,h)f(t,dx,dv,ds,dh)
\\
&=\int\!\!\!\iiint_{Z}\bE[\chi(T_t(x,v,s,h)))]f^{in}(x,\om,\tau,h)dxdvdsdh\,,
\end{aligned}
$$
where $\bE$ designates the expectation on trajectories of the sequence of i.i.d. random variables 
$\bbb=(b_n)_{n\ge 1}$. 

In other words, $f(t,\cdot,\cdot,\cdot,\cdot)$ is the image under the map $T_t$ of the measure 
$\hbox{Prob}(d\bbb)f^{in}(x,\om,\tau,h)$, where 
$$
\hbox{Prob}(d\bbb)=\prod_{n\ge 1}dm(b_n)\,.
$$

Set $g(t,x,v,s,h)=\bE[\chi(T_t(x,v,s,h)))]$; one has
$$
g(t,x,v,s,h)=\bE[\indc_{t<s}\chi(T_t(x,v,s,h)))]+\bE[\indc_{s<t}\chi(T_t(x,v,s,h)))]\,.
$$
If $s>t$, there is no collision in the time interval $[0,t]$ for the trajectory considered, meaning that
$$
T_t(x,v,s,h)=(x+tv,v,s-t,h)\,.
$$
Hence
$$
\bE[\indc_{t<s}\chi(T_t(x,v,s,h)))]=\chi(x+tv,v,s-t,h)\indc_{t<s}\,.
$$

On the other hand
$$
\ba
\bE[\indc_{s<t}\chi(T_t(x,v,s,h))]
=
\bE[\indc_{s<t}\chi(T_{(t-s)-0}T_{s+0}(x,v,s,h))]
\\
=
\bE[\indc_{s<t}\chi(T_{(t-s)-0}(x+sv,\cR[\Dlt(h)]v,s_1,h_1))]
\ea
$$
with $(s_1,h_1)=\bT_{b_1}(h)$ and $\Dlt(h)=2\arcsin(h)-\pi$. 

Conditioning with respect to $(s_1,h_1)$ shows that
$$
\ba
\bE&[\indc_{s<t}\chi(T_t(x,v,s,h)))]
\\
&=
\bE[\indc_{s<t}
\bE[\chi(T_{(t-s)-0}(x+sv,\cR[\Dlt(h)]v,s_1,h_1))|s_1,h_1]]\,,
\ea
$$
and
$$
\ba
\bE[\chi(T_{(t-s)-0}(x+sv,\cR[\Dlt(h)]v,s_1,h_1))|s_1,h_1]
\\
=
g(t-s,x+sv,\cR[\Dlt(h)]v,s_1,h_1)\,.
\ea
$$
Then
$$
\ba
\bE[\indc_{s<t}
\bE[\chi(T_{(t-s)-0}(x+sv,\cR[\Dlt(h)]v,s_1,h_1))|s_1,h_1]]
\\
=\indc_{s<t}
	\int g(t-s,x+sv,\cR[\Dlt(h)]v,\bT_{b_1}(h))]dm(b_1)
\\
=\indc_{s<t}
	\int g(t-s,x+sv,\cR[\Dlt(h)]v,s_1,h_1)]P(s_1,h_1|h)ds_1dh_1\,.
\ea
$$
Finally
$$
\ba
g(t,x,v,s,h)&=\chi(x+tv,v,s-t,h)\indc_{t<s}
\\
&+
\indc_{s<t}
	\int g(t-s,x+sv,\cR[\Dlt(h)]v,s_1,h_1)]P(s_1,h_1|h)ds_1dh_1\,.
\ea
$$
This formula represents the solution of the problem
$$
\ba
(\d_t-v\cdot\grad_x+\d_s)g&=0\,,
	\quad t,s>0\,,\,\,x\in\bR^2\,,\,\,s\in\bS^1\,,\,\,|h|<1
\\
g(t,x,s,0,h)&=
	\iint_{\bR^*_+\times]-1,1[}P(s_1,h_1|h)g(t,x,v,s_1,h_1)ds_1dh_1\,,
\\
g\rstr_{t=0}&=\chi\,.
\ea
$$
The boundary condition for $s=0$ can be replaced with a source term that is proportional to the 
Dirac measure $\de_{s=0}$:
$$
\ba
(\d_t-v\cdot\grad_x+\d_s)g&=\de_{s=0}\iint_{\bR^*_+\times]-1,1[}
	P(s_1,h_1|h)g(t,x,v,s_1,h_1)ds_1dh_1\,,
\\
g\rstr_{t=0}&=\chi\,.
\ea
$$
One concludes by observing that this problem is precisely the adjoint of the Cauchy problem in the 
theorem.

\bigskip
Let us conclude this section with a few bibliographical remarks. 

Although the Boltzmann-Grad limit of the periodic Lorentz gas is a fairly natural problem, it remained
open for quite a long time after the pioneering work of G. Gallavotti on the case of a Poisson distribution
of obstacles \cite{Gallavotti1969, Gallavotti1972}. 

Perhaps the main conceptual difficulty was to realize that this limit must involve a phase-space other 
than the usual phase-space of kinetic theory, i.e. the set $\bR^2\times\bS^1$ of particle positions and 
velocities, and to find the appropriate extended phase-space where the Boltzmann-Grad limit of the 
periodic Lorentz gas can be described by an autonomous equation.  

Already Theorem 5.1 in \cite{CagliotiFG2003} suggested that, even in the simplest problem where 
the obstacles are absorbing --- i.e. holes where particles disappear forever, --- the limit of the particle 
number density in the Boltzmann-Grad scaling cannot described by an autonomous equation in the
usual phase space $\bR^2\times\bS^1$. 

The extended phase space $\bR^2\times\bS^1\times\bR_+\times[-1,1]$ and the structure of the 
limit equation were proposed for the first time by E. Caglioti and the author in 2006, and presented 
in several conferences --- see for instance \cite{Golse2007};  the first computation of the transition 
probability $P(s,h|h')$ (Theorem \ref{T-TransiProba}), together with the limit equation (Theorem 
\ref{T-LimitEq}) appeared in \cite{CagliotiFG2008} for the first time.  However, the theorem 
concerning the limit equation in \cite{CagliotiFG2008} remained incomplete, as it was based 
on the independence assumption (H). 

Shortly after that, J. Marklof and A. Str\"ombergsson proposed a complete derivation of the limit 
equation of Theorem \ref{T-LimitEq} in a recent preprint \cite{MarkloStrom2008}. Their analysis,
establish the validity of this equation in any space dimension, using in particular the existence
of a transition probability as in Theorem \ref{T-TransiProba} in any space dimension, a result
that they had proved in an earlier paper \cite{MarkloStrom2007}. The method of proof in this 
article \cite{MarkloStrom2007} avoided using continued or Farey fractions, and was based on
group actions on lattices in the Euclidian space, and on an important theorem by M. Ratner
implying some equidistribution results in homogeneous space. However, explicit computations 
(as in Theorem \ref{T-TransiProba} of the transition probability in space dimension higher than 
$2$ seem beyond reach at the time of this writing --- see however \cite{MarkloStrom3} for
computations of the 2-dimensional transition probability for more general interactions than
hard sphere collisions. 

Finally, the limit equation obtained in Theorem \ref{T-LimitEq} is interesting in itself; some
qualitative properties of this equation are discussed in \cite{CagliotiFG2009}.


\section*{Conclusion}


Classical kinetic theory (Boltzmann theory for elastic, hard sphere collisions) is based on two 
fundamental principles

a) deflections in velocity at each collision are mutually independent and identically distributed

b) time intervals between collisions are mutually independent, independent of velocities,  and 
exponentially distributed.

\smallskip
The Boltzmann-Grad limit of the periodic Lorentz gas provides an example of a non classical 
kinetic theory where 

a') velocity deflections at each collision jointly form a Markov chain;

b') the time intervals between collisions are not independent of the velocity deflections.

\smallskip
In both cases, collisions are purely local and instantaneous events: indeed the Boltzmann-Grad 
scaling is such that the particle radius is negligeable in the limit. The difference between these 
two cases is caused by the degree of correlation between obstacles, which is maximal in the 
second case since the obstacles are centered at the vertices of a lattice in the Euclidian space,
wheras obstacles are assumed to be independent in the first case. It could be interesting to
explore situations that are somehow intermediate between these two extreme cases --- for
instance, situations where long range correlations become negligeable.

\smallskip
Otherwise, there remain several outstanding open problems related to the periodic Lorentz
gas, such as

\smallskip
i) obtaining explicit expressions of the transition probability whose existence is proved by 
J. Marklof and A. St\"rombergsson in \cite{MarkloStrom2007}, in all space dimensions, or 

ii) treating the case where particles are accelerated by an external force --- for instance the 
case of a constant magnetic field, so that the kinetic energy of particles remains constant.



\end{document}